\documentclass{amsart}

\usepackage{amssymb}
\usepackage{amssymb,amsmath,amscd}
\usepackage{pstricks,pst-node}
\psset{unit=.5cm}

\usepackage{graphicx}
\usepackage{epstopdf}
\usepackage{amscd}
\usepackage{graphics}
\newtheorem{Theorem}{Theorem}[section]
\newtheorem{Proposition}{Proposition}[section]
\newtheorem{Corollary}{Corollary}[section]
\newtheorem{Lemma}{Lemma}[section]

\def\proof{\par{\it Proof}. \ignorespaces}
\def\endproof{{\ \vbox{\hrule\hbox{%
     \vrule height1.3ex\hskip0.8ex\vrule}\hrule }}\par}
\newenvironment{Proof}{\proof}{\endproof}

\theoremstyle{definition}
\newtheorem{Definition}[Theorem]{Definition}

\theoremstyle{remark}
\newtheorem{Remark}[Theorem]{Remark}

\numberwithin{equation}{section}

\numberwithin{figure}{section}

\let\trueint=\int
\let\truesum=\sum
\def\int{\mathop{\textstyle\trueint}\limits}
\def\sum{\mathop{\textstyle\truesum}\limits}

\let\<=\langle
\let\>=\rangle

\renewcommand\labelitemi{\ifmmode\circ\else$\circ$\fi}

\begin{document}


\title[Dispersionless integrable systems and universality in random matrix theory]
{Combinatorics of dispersionless integrable systems and universality in random matrix theory}

\author{Yuji Kodama}
\address{Department of Mathematics, Ohio State University, Columbus,
OH 43210}
\email{kodama@math.ohio-state.edu}

\thanks{Both authors are partially
supported by NSF grant DMS0806219}

\author{Virgil U. Pierce}
\address{Department of Mathematics, University of Texas -- Pan American, Edinburg, TX  78539}
\email{piercevu@utpa.edu}

\begin{abstract}
It is well-known that the partition function of the unitary ensembles of random matrices
 is given by a $\tau$-function
of the Toda lattice hierarchy and those of the orthogonal and symplectic ensembles are
$\tau$-functions of the Pfaff lattice hierarchy.
In these cases the asymptotic expansions of the free energies given by the logarithm of the partition functions lead to the dispersionless (i.e. continuous) limits for the Toda and Pfaff lattice hierarchies.
There is a universality between all three ensembles of random matrices, one consequence of which is that the leading orders of the
free energy for large matrices agree.  In this paper, this universality, in the case of Gaussian ensembles, is explicitly demonstrated by computing the leading orders of the free energies in the expansions.
We also show that the free energy as the solution of the dispersionless Toda lattice hierarchy gives
a solution of the dispersionless Pfaff lattice hierarchy, which implies that this universality holds
in general for the leading orders of the unitary, orthogonal, and symplectic ensembles.

We also find an explicit formula for the two point function $F_{nm}$ which
represents the number of connected ribbon graphs with two vertices of degrees $n$ and $m$
on a sphere. The derivation is based on the Faber polynomials defined on the spectral curve of the dispersionless Toda lattice hierarchy, and $\frac{1}{nm}F_{nm}$ are the Grunsky coefficients of the Faber polynomials.
\end{abstract}

\maketitle

\thispagestyle{empty}
\pagenumbering{roman}\setcounter{page}{1}
\tableofcontents

\pagenumbering{arabic}
\setcounter{page}{1}
\setcounter{figure}{1}

\section{Introduction and Background}
We begin with a brief review of the dispersionless limits of the Toda and Pfaff lattice hiearchies (see for examples \cite{AvM:02, KP:07, TT:95, T:07}). Here the Toda hierarchy is the 1-dimensional one with the parameters $\{t_0,t_1,t_2,\ldots\}$.  These limits are found from the equations for the $\tau$-functions of the lattice hierarchies by subsituting an asymptotic expansion for the logarithm of the $\tau$-functions in the small $\hbar$ limit.  We also describe the connection of the $\tau$-functions of these hierarchies with the partition functions of the unitary, orthogonal, and symplectic ensembles of random matrix models
with the matrix size $N$ (or $2N$) related to $N=\frac{1}{\hbar}$ (see for examples \cite{GHJ, EM-2003, MW-2003}).  In particular the asymptotic expansions of the free energy, the logarithm of the partition functions, as the size of the matrices approaches infinity satisfy the underlying assumption for the asymptotic expansion of the logarithm of the $\tau$-functions  of the Toda and Pfaff lattice hierarchies.  Then the leading order terms of the asymptotic expansions satisfy the dispersionless Toda and Pfaff lattice hierarchies, which are written in terms of
the two-point functions denoted by $F_{nm}$ (see for examples \cite{CK:95, TT:95, CT:06}).
There is a rich literature dealing with dispersionless integrable systems including the KP and Toda lattice hierarchies, which mainly involve studying algebraic and complex analytic aspects of the equations and/or discussing applications to 2-dimensional topological field theories (see, for examples, \cite{AK:96, D:92, Kr:94, MWZ:02, WZ:00, LT:03, Z:01}). However, it seems that there is no paper which connects directly the dispersionless integrable systems with combinatorial problems associated with random matrix
ensembles.
The present paper is to deal with this connection for
combinatorial problems such as counting ribbon graphs on a compact surface of genus zero, and
 give a universality result for the dispersionless integrable systems which is an analogue of that for the random matrix ensembles.

\subsection{The Toda lattice hierarchy}
The Toda lattice equation (the first member of the Toda lattice hierarchy) is given by
\[\left\{
\begin{array}{lllll}
\displaystyle{\frac{\partial a_n}{\partial t_1}}&=a_n(b_{n+1}-b_n),\\[2.0ex]
\displaystyle{\frac{\partial b_n}{\partial t_1}}&=a_n-a_{n-1}\,.
\end{array}\right. \qquad n=1,2,\ldots.
\]
There exists a sequence $\{\tau_n:n\ge 0\}$ of $\tau$-functions with $\tau_0=1$ which generate the $a_n, b_n$ by the formulas
\[
a_n=\frac{\tau_{n+1}\tau_{n-1}}{\tau_n^2},\qquad b_n=\frac{\partial}{\partial t_1}\log \left(\frac{\tau_n}{\tau_{n-1}}\right)\,.
\]
We may then write the Toda lattice equation in the Hirota bilinear form,
\begin{equation}\label{Toda-Hirota}
D_1^2\tau_n\cdot\tau_n=2\tau_{n+1}\tau_{n-1}.
\end{equation}
Here $D_1$ is the usual Hirota derivative, i.e. for a variable $t_k$ which is a flow-parameter for the $k$-th member of the Toda lattice hierarchy, $D_k$ is defined by
\[
D_kf\cdot g:=\left(\frac{\partial}{\partial t_k}-\frac{\partial }{\partial t_k'}\right) f(t_k)g(t_k')\Big|_{t_k=t_k'}\,.
\]
Note that the Hitota equation (\ref{Toda-Hirota}) of the Toda lattice gives
\begin{equation}\label{an}
a_n=\frac{\partial^2}{\partial t_1^2}\log\tau_n=\frac{\tau_{n+1}\tau_{n-1}}{\tau_{n}^2}\,.
\end{equation}

The Toda lattice is expressed in the Lax representation with a tridiagonal semi-infinite matrix $L$ by
\[
\frac{\partial L}{\partial t_1}=[L, B_1],\qquad B_1=[L]_{-}\,,
\]
where $[L]_-$ is the strictly lower triangular part of the matrix $L$ given by
\[
L=\begin{pmatrix}
b_1 &  1 &  0  & 0 & \cdots  \\
a_1 & b_2 & 1 & 0 & \cdots \\
0  & a_2 & b_3 & 1 & \cdots\\
0 &  0 & a_3 & b_4 & \cdots \\
\vdots&\vdots &\vdots&\vdots& \ddots
\end{pmatrix}\,.
\]
The Lax matrix is also written in terms of a shift operator $\Delta:=\exp(\frac{\partial}{\partial n})$, that is,
with the eigenvector $\phi=(\phi_1,\phi_2,\ldots)$, we have
\begin{align*}
(L\phi)_n &=\phi_{n+1}+b_n\phi_n+a_{n-1}\phi_{n-1}\\
&=(\Delta + b_n + a_{n-1}\Delta^{-1})\phi_n =\lambda\phi_n.
\end{align*}
The hierarchy of the Toda lattice is defined by
\[
\frac{\partial L}{\partial t_k}=[L, B_k],\qquad  B_k=[L^k]_-,\qquad k=1,2,3,\ldots\,.
\]
The $\tau$-functions are then the functions of infinite variables, i.e. $\tau(\mathbf{t})$ with $\mathbf{t}=(t_1,t_2,\ldots)$.
Then it is well-known that each $\tau_n$-function also satisfies the KP hierarchy \cite{AvM:02, BK:03}:
\begin{equation}\label{KPhierarchy}
\left(h_{k+1}(\hat{\mathbf{D}})-\frac{1}{2}D_1D_{k}\right)\tau_n\cdot\tau_n=0,
\qquad k=3,4,5,\ldots.
\end{equation}
where $\hat{\mathbf{D}}=(D_1,\frac{1}{2}D_2,\frac{1}{3}D_3,\ldots)$ and $h_n(\mathbf{x})$ with $\mathbf{x}=\hat{\mathbf{D}}$ are the elementary symmetric functions of $\mathbf{x}=(x_1,x_2,\ldots)$ defined by
\[
\exp\left(\sum_{k=1}^{\infty}x_n z^n\right)=\sum_{n=0}^{\infty}h_n(\mathbf{x})z^n.
\]
In particular, the first equation with $k=3$ gives
\[
(-4D_1D_3+D_1^4+3D_2^2)\tau_n\cdot\tau_n=0,
\]
which is the KP equation, that is,
the function $u=2\frac{\partial^2}{\partial x^2}\log \tau_n$ for each $n$ satisfies
\begin{equation}\label{KP}
\frac{\partial}{\partial x}\left(-4\frac{\partial u}{\partial t}+\frac{\partial^3u}{\partial x^3}+6u\frac{\partial u}{\partial x}\right)+3\frac{\partial^2u}{\partial y^2}=0\,,
\end{equation}
with $x=t_1, y=t_2$ and $t=t_3$.

It is also well-known that the $\tau_n$-functions of the Toda lattice hierarchy satisfy the following set of equations (see for example \cite{AvM:99}):
\begin{equation}\label{dHT2}
\left(D_{k}-h_k(\hat{\mathbf{D}})\right)\tau_{n+1}\cdot\tau_n=0,\qquad k=2,3,4,\ldots.
\end{equation}
The first equation with $k=2$ gives
\begin{equation}\label{nlsh}
(D_2-D_1^2)\tau_{n+1}\cdot\tau_n=0\,.
\end{equation}
This equation with (\ref{Toda-Hirota}) gives the nonlinear Schr\"odinger equation (with the change $t_2\to it_2$), i.e.
\[
i\frac{\partial \psi}{\partial t_2}+\frac{\partial^2\psi}{\partial t_1^2}+2\psi^2\bar{\psi}=0\,,
\]
with $\psi=\frac{\tau_{n+1}}{\tau_n}$ and $\bar{\psi}=\frac{\tau_{n-1}}{\tau_n}$. Thus, the nonlinear Schr\"odinger equation is the second member of the Toda lattice hierarchy.

We now briefly summarize the dispersionless limit of the Toda hierarchy:
The key ingredient in the dispersionless limit is to introduce a free energy, denoted $F$,
\begin{equation}\label{asymptotic-assumption}
\tau_n(\mathbf{t}; \hbar)=\exp\left(\frac{1}{\hbar^2}F(T_0,\mathbf{T}) +O(\hbar^{-1})\right),
\end{equation}
where $\hbar$ is a small parameter, and $ \mathbf{T}=(T_1,T_2,\ldots)$ represents the slow variables with $T_k=\hbar t_k$ for $n\ge 1$.
Note in particular that  the limit $\hbar\to 0$ gives a continuous limit of the lattice structure, that is,
the lattice spacing has the order $O(\hbar)$ and the limit introduces the continuous variable $T_0$,
\[
\hbar n \longrightarrow T_0.
\]
The free energy $F(T_0,\mathbf{T})$ is calculated from the limit,
\[
F(T_0,\mathbf{T})=\lim_{\hbar\to 0}\,\hbar^2\log\left[\tau_n\left(\hbar^{-1}\mathbf{T};\hbar\right)\right].
\]
Then the dispersionless Toda hierarchy can be written in terms of the second derivatives,
\[
F_{mn}=\frac{\partial^2 F}{\partial T_m\partial T_n}\,,\qquad m,n \ge 0.
\]
The $F_{mn}$ play an important role for the dispersionless integrable systems, and they are sometimes referred to as the {\it two point functions} of the corresponding topological field theory (see for examples \cite{D:92, Kr:94, AK:96}).

In the limit $\hbar\to 0$, the $a_n,b_n$ variables in the Toda lattice become
\[\left\{
\begin{array}{llll}
\displaystyle{a_n(t)=\exp\left(\log\tau_{n+1}-2\log\tau_n+\log\tau_{n-1}\right)\longrightarrow e^{F_{00}}},\\[1.5ex]
\displaystyle{b_n(t)=\frac{\partial}{\partial t_1}(\log\tau_{n}-\log\tau_{n-1})\longrightarrow F_{01}} \,.
\end{array}\right.
\]
Then the Hirota form of the Toda lattice equation (\ref{Toda-Hirota}) becomes
\begin{equation}\label{dToda-1}
F_{11}=e^{F_{00}}\,,
\end{equation}
which is called the dispersionless Toda (dToda) equation. Likewise the KP equation (\ref{KP}) gives
\[
-4F_{13}+6F_{11}^2+3F_{22}=0\,,
\]
which is the dispersionless KP (dKP) equation. Also the dispersionless limit of \eqref{nlsh} is
\[
F_{02}-2F_{11}-F_{01}^2=0\,.
\]

One should also note that in the dispersioless limit, the spectral problem $L\phi=\lambda\phi$ with the Lax operator and the eigenvector $\mathbf{\phi}=(\phi_1,\phi_2,\ldots)^T$ leads to an algebraic equation
(the spectral curve of the dToda equation),
\begin{equation}\label{todacurve}
\lambda=p(\lambda)+F_{01}+\frac{e^{F_{00}}}{p(\lambda)}\,,
\end{equation}
where $p(\lambda)$ is a quasi-momentum given by $\lim_{\hbar\to 0}\Delta\phi/\phi=\exp(\partial S/\partial  T_0)$. This is a spectral curve of genus $0$, and the dToda hierarchy defines a integrable deformation
of the curve (see Lemma \ref{tridiagonality}). Here the eigenvector $\phi$ is assumed to be in the WKB form as $\hbar\to 0$,
\[
\phi_n(\mathbf{t};\hbar)=\exp\left(\frac{1}{\hbar}S(T_0,\mathbf{T})+O(1)\right).
\]
Although the function $S$ plays a key role for the dispersionless theory, we will not use it in this paper.
We just mention that the $S$-function can be expressed in the form (see for example \cite{T:07}),
\begin{equation}\label{S}
S(T_0,\mathbf{T})=\sum_{n=1}^{\infty}\lambda^nT_n+T_0\log\lambda-\sum_{k=1}^{\infty}\frac{1}{k\lambda^k}F_k(T_0,\mathbf{T})\,.
\end{equation}

The dToda hierarchy may be formulated by the form,
\begin{equation}\label{dToda-Fay}\left\{
\begin{array}{lllll}
\displaystyle{e^{D(\lambda)D(\mu)F}=\frac{\lambda e^{-D(\lambda)F_0}-\mu e^{-D(\mu)F_0}}{\lambda-\mu}}\\[2.0ex]
\displaystyle{e^{-D(\lambda) D(\mu)F}=1-\frac{e^{F_{00}}}{\lambda\mu}e^{(D(\lambda)+D(\mu))F_0}}
\end{array}\right.
\end{equation}
where
\[
D(\lambda)=\sum_{n=1}^{\infty}\frac{1}{n\lambda^n}\frac{\partial}{\partial T_n}.
\]
This is obtained by the reduction from the 2-dimesional dToda hierarchy in \cite{CT:06,T:07}
with the constraints $\bar t_k=-t_k$ for all $k$, and the Toda hierarchy with those constraints
is referred to as the 1-dimensional Toda hierarchy.  Then one can show that the dToda hierarchy
\eqref{dToda-Fay} leads to the curve \eqref{todacurve} (see Lemma \ref{tridiagonality} below).
We will further discuss the detailed structure of the hierarchy \eqref{dToda-Fay} in Section
\ref{S:dToda}, where we show that \eqref{dToda-Fay} indeed gives the dispersionless limits
of \eqref{KPhierarchy} and \eqref{dHT2}.


\subsection{Unitary ensembles of random matrices}\label{UnitaryToda}
A special class of solutions to the Toda lattice equations defines the partition functions of the unitary ensembles of random matrices.  Explicitly
 the $\tau$-functions are taken to be
\[
\tau_n(\mathbf{t}) = Z_n^{(2)}(V_0(\lambda); \mathbf{t}) =\int_{\mathbb{R}}{d\lambda_1}\cdots\int_{\mathbb{R}}d\lambda_n\,\prod_{i<j} |\lambda_i - \lambda_j|^2 \exp\left( - \sum_{k=1}^n V_{\mathbf{t}}(\lambda_k) \right)\,,\]
where
$ V_{\mathbf{t}}(\lambda_k)=V_0(\lambda_k)-\sum_{j=1}^{\infty}\lambda_k^jt_j$ and $V_0(\lambda)$ is a polynomial of even degree  (e.g. $V_0(\lambda)=\frac{1}{2}\lambda^2$ for the Gaussian ensemble).
This function represents integration over the eigenvalues of the random matrices of the unitary ensemble, and its Taylor expansion in $\mathbf{t}$ contains the moments of traces of powers of the matrix (see for example a common reference \cite{M:2004}).

A ribbon graph is a one complex embedded on an oriented genus $g$ surface,
such that the complement of the complex is a disjoint collection of sets
homeomorphic to discs.   We may think of a ribbon graph as a collection of
discs (vertices) and ribbons (edges) such that the ribbons are glued together
preserving the local orientations of the discs in other words the ribbons lay
flat on the oriented surface.  Our ribbon graphs will be labeled in the
following way:  The vertices are labeled so that they are distinct, and the
edges emnating from each vertex are labeled as distinct as well. We also
define the degree of the vertex as the umber of edges attached to
the vertex. We represent ribbon graphs for this paper as graphs drawn on an oriented surface (we only consider a compact surface of genus zero (i.e. a sphere) in this paper, and refer a ribbon graph on the sphere as a ribbon graph of genus zero).
In Figure \ref{fig:vertex}, we illustrate one vertex of degree 8 and two vertices of degrees
of 3 and 5 on a sphere.
\begin{figure}
\begin{center}
\includegraphics[width=10cm]{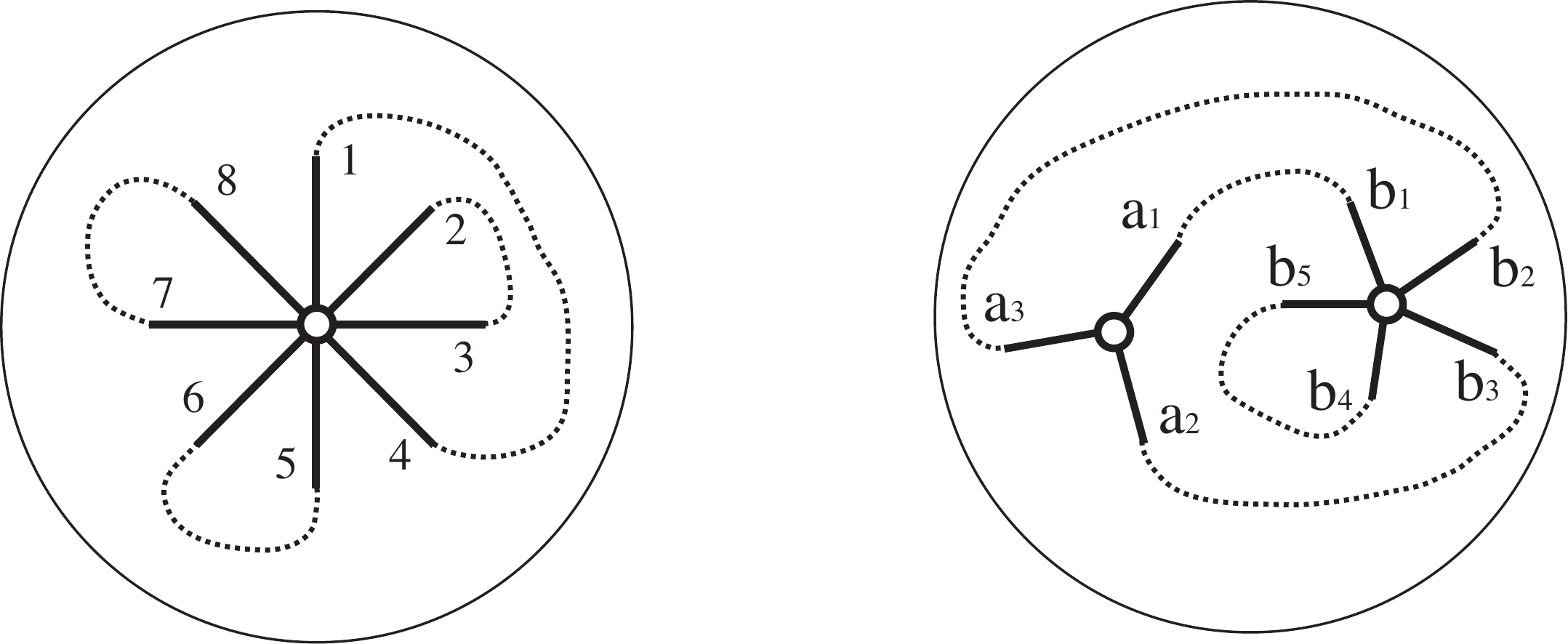}
\caption{A (ribbon) graph with  one oriented vertex of degree 8 on a sphere (left) and a graph with two oriented vertices of degrees 3 and 5 on a sphere (right).
Examples of complete connected
graphs are shown.  The edges drawn are to represent ribbons laying flat on the sphere.
\label{fig:vertex}}
\end{center}
\end{figure}

A fundamental result of random matrix theory is that for a quadratic potential $V_0(\lambda)$, $\log( Z^{(2)}_N)$ posseses an asymptotic expansion in even powers of ${N}$, whose terms give generating functions partitioning ribbon graphs by genus:
\begin{Theorem}[Bessis-Itzykson-Zuber \cite{BIZ}, Ercolani-McLaughlin \cite{EM-2003}] \label{BIZ-thm}
For the Gaussian unitary ensemble (GUE), we have
\[ \log\left[ Z_N^{(2)}\left( \frac{N}{2} \lambda^2; N \mathbf{T}\right) \Big/ Z_N^{(2)}\left(\frac{N}{2} \lambda^2; \mathbf{0}\right) \right] =
\sum_{g\geq 0} e_g(\mathbf{T}) N^{2-2g} \]
where
\[ e_g(\mathbf{T})
=\sum_{0\le j_1,j_2,\ldots}\kappa_g(j_1,j_2,\dots)\,\frac{T_1^{j_1}T_2^{j_2}\cdots}{j_1!j_2!\cdots}\,=\sum_{\mathbf{j}} \kappa_g(\mathbf{j}) \frac{\mathbf{T}^\mathbf{j}}{\mathbf{j}!} \,. \]
The coefficient $\kappa_g(\mathbf{j})$ gives the number of the connected ribbon graphs with $j_k$ labeled vertices of degree $k$ for $k=1,2,\ldots$ on a compact surface of genus $g$.
\end{Theorem}
Note that Theorem \ref{BIZ-thm} does not contain our $T_0 = \frac{n}{N} = n \hbar$ variable
in large $N$ limit.
The $T_0$-variable can be inserted naturally by taking the $\tau$-function in the form
with $\mathbf{t}=N\mathbf{T}$ (or $\mathbf{T}=\hbar\mathbf{t}$),
\[
\tau_n(\mathbf{t};\hbar)=Z_n^{(2)}\left( \frac{N}{2} \lambda^2 ; N \mathbf{T}\right)\,.
\]
Then the free energy is given by
\[
\log\left[\tau_n(\hbar^{-1}\mathbf{T};\hbar)\right]=\log\left[ Z_n^{(2)}\left( \frac{N}{2} \lambda^2; N \mathbf{T}\right) \Big/ Z_n^{(2)}\left(\frac{N}{2} \lambda^2; \mathbf{0}\right) \right] +\log\left[Z_n^{(2)}\left(\frac{N}{2} \lambda^2; \mathbf{0}\right)\right]\,.
\]
We note here that the scaling $\lambda\to\sqrt{T_0}\lambda$ leads to
\[
 Z_n^{(2)}\left( \frac{N}{2} \lambda^2 ; N \mathbf{T}\right)  = Z_n^{(2)}\left(  \frac{n}{2 T_0} \lambda^2;  \frac{n}{T_0} \mathbf{T} \right)
= T_0^{\frac{n^2}{2}}Z_n^{(2)}\left( \frac{n}{2} \lambda^2; n \mathbf{\hat{T}} \right)\,,
\]
where $\hat{T}_j = T_0^{j/2 -1} T_j $ (which is the so called Penner scaling).
This expression then has the asymptotic expansion including $T_0$,
\begin{align*}
\log\left[\tau_n(\hbar^{-1}\mathbf{T};\hbar)\right]&
= \sum_{g\geq 0} e_g( \mathbf{\hat{T}} ) n^{2-2g} +C(T_0;N)\\
&= \sum_{g\geq 0} T_0^{2-2g} e_g(\mathbf{\hat{T}} )\, N^{2-2g}  + C(T_0;N)\,,
\end{align*}
where $C(T_0;N)=\log[Z_n^{(2)}(\frac{N}{2}\lambda^2;\mathbf{0})]$ which can be computed directly (see Appendix A).
Note that this asymptotic expansion agrees (at leading order) with the assumption (\ref{asymptotic-assumption}); and thus gives a class of solutions of the Toda lattice satisfying this assumption.  Then
the leading order of this asymptotic expansion gives a solution of the dispersionless
Toda equation,
\begin{align*}
F(T_0, \mathbf{T}) &=\lim_{N\to\infty}\,\frac{1}{N^2}\,\log\left[Z_n^{(2)}\left(\frac{N}{2}\lambda^2;N\mathbf{T}\right)\right]\\
&= T_0^2 e_0( \hat{\mathbf{T}}) + C_0(T_0)\,,
\end{align*}
where $C_0(T_0)$ is the leading order of $C(T_0;N)$ for large limit of $N$, and
as shown in Appendix B, it is given by
\begin{equation}\label{C0Toda}
C_0(T_0)=\frac{1}{2}T_0^2\left(\log T_0-\frac{3}{2}\right)\,.
\end{equation}
This formula has also been found directly from the dToda equations in a connection to
the topological field theory (see (8.6) of p.209 in \cite{AK:96}).

In particular, Theorem \ref{BIZ-thm} implies that the two point function $F_{nm}(1; \mathbf{0})$
for $mn\ne0$
represents the number of connected ribbon graphs with two vertices of degrees $n$ and $m$
on a sphere $(g=0)$, i.e.
\[
F_{nm}(1;\mathbf{0})=\frac{\partial^2 F}{\partial T_n\partial T_m}(1;\mathbf{0})=
\frac{\partial^2 e_0}{\partial T_n\partial T_m}(\mathbf{0})=\kappa_0
(0,\ldots,1,0,\ldots,1,0,\ldots), \qquad n, m \ge1\,,
\]
where $1$'s in $\kappa_0(\mathbf{j})$ are at the $n$- and $m$-th places. We will give an explicit formula for $F_{nm}(1;\mathbf{0})$ in Section \ref{combin-dToda}.


We also attach an enumerative meaning to derivatives with respect to
$T_0$:
\begin{Corollary}\label{comb-F0}
Each derivative $F_{0,2k}$ for $k\ge 1$, corresponds to counting the number of connected ribbon graphs
with a vertex of degree $2k$ and a marked face on a sphere.
\end{Corollary}
\begin{proof}
We start with our relation
\begin{align}
F(T_0; \mathbf{T}) &= T_0^2 e_0(\hat{\mathbf{T}})+C_0(T_0) \label{T0meaning}\\
&= \sum_{\mathbf{j}} \kappa_0(\mathbf{j}) T_0^2
\frac{\hat{\mathbf{T}}^\mathbf{j} }{\mathbf{j}!} +C_0(T_0)\nonumber \\
&= \sum_{\mathbf{j}} \kappa_0(\mathbf{j}) T_0^{2 - \sum_i j_i + \frac{1}{2}\sum_i i j_i }
\frac{\mathbf{T}^\mathbf{j}}{\mathbf{j}! } +C_0(T_0).\nonumber
\end{align}
One then notes that the number of vertices of a ribbon graph with vertices of
type $\mathbf{j} = (j_1, j_2, \dots)$ is $v = \sum_i j_i $ and the number of edges is $e = \frac{1}{2}\sum_i
i j_i $.  If the ribbon graph is of genus 0, then by the Euler
characteristic formula, the number of faces is
\begin{equation*}
f = 2 - v + e = 2 - \sum_i j_i + \frac{1}{2}\sum_i j_i i
\end{equation*}
so that by inserting this equation into (\ref{T0meaning}), we find
\begin{equation}\label{Ufree}
 F(T_0; \mathbf{T}) = \sum_{\mathbf{j}} \kappa_0(\mathbf{j}) T_0^f
\frac{\mathbf{T}^\mathbf{j}}{\mathbf{j}!} +C_0(T_0)\,,
\end{equation}
from which one sees that differentiating with respect to $T_0$ is related to
counting ribbon graphs with marked faces in the enumerative part (i.e. except
for the contributions coming from $C_0$).  For example:
$F_{0, 2k}(1; \mathbf{0})$ with $k\ne 0$ is the number of ribbon graphs with a single vertex
of degree $2k$ and one of the $f= 1 + k$ faces marked.
These numbers will also be found in Theorem \ref{Fformula}
as a consequence of the dToda hierarchy.
\end{proof}


\subsection{The Pfaff lattice hierarchy}
The Pfaff lattice hierarchy is normally referred to as the DKP hierarchy (The ``D" stands for
the D-type Lie group which is the symmetry group for the hierarchy \cite{JM:83}).
The first member of the DKP hierarchy, called the DKP equation,
 is given by the set of equations
\begin{equation}\label{DKPequation}\left\{\begin{array}{lllll}
\displaystyle{\frac{\partial}{\partial t_1}\left(-4\frac{\partial u}{\partial t_3}+\frac{\partial^3u}{\partial t_1^3}+12u\frac{\partial u}{\partial t_1}\right)+3\frac{\partial^2u}{\partial t_2^2}=12\frac{\partial^2}{\partial t_1^2}(v^+v^-)}\\[2.5ex]
\displaystyle{2\frac{\partial v^{\pm}}{\partial t_3}+\frac{\partial^3 v^{\pm}}{\partial t_1^3}+6u\frac{\partial v^{\pm}}{\partial t_1}
\mp 3\left(\frac{\partial^2 v^{\pm}}{\partial t_1\partial t_2}+2 v^{\pm}\int^{t_1}\frac{\partial u}{\partial t_2}\,dt_1\right)=0}  \,.
\end{array}\right.
\end{equation}
The left hand side of the first equation is just the KP equation, and the right hand side
gives a coupling term with the field $v^{\pm}$. Because of this, the DKP equation is
sometimes called a coupled KP equation \cite{HO:91}.

In terms of the $\tau$-functions, $u$ and $v^{\pm}$ are expressed as
\begin{equation}\label{uv}
u=\frac{\partial^2}{\partial t_1^2}\log \tau_n,\qquad v^{\pm}= \frac{\tau_{n\pm 1}}{\tau_n}\,.
\end{equation}
The DKP equation is then given by
\begin{equation}\label{Pfaff}
\left\{\begin{array}{lll}
\displaystyle{(-4D_1D_3+D_1^4+3D_2^2)\tau_n\cdot\tau_n=24 \tau_{n+1}\tau_{n-1}}\\[1.5ex]
\displaystyle{(2D_3+D_1^3\mp 3D_1D_2)\tau_{n\pm 1}\cdot\tau_{n}=0}\,.
\end{array}\right.
\end{equation}
It is known that the $\tau$-functions are given by Pfaffians \cite{HO:91, ASvM:02, KP:07}.
Since the equation with the set of $\tau$-functions has a lattice structure (i.e. $\tau_{n}$ is determined by the  previous $\tau$-functions with $\tau_0=1$),
we call it the {\it Pfaff lattice} in this paper.  This choice of name is also convenient for this paper as we wish to compare the Pfaff lattice, as given in \cite{ASvM:02, KP:07}, and the Toda lattice.  The Pfaff lattice hierarchy is then given by \cite{AvM:02, T:07}
\[\left\{\begin{array}{llll}
\left(h_{k+4}(\tilde{\mathbf{D}})-\frac{1}{2}D_1D_{k+3}\right)\tau_n\cdot\tau_n=h_{k}(\tilde{\mathbf{D}})\,\tau_{n+1}\cdot\tau_{n-1},\\[2.0ex]
\quad h_{k+3}(\mp \tilde{\mathbf{D}})\,\tau_{n\pm 1}\cdot\tau_{n}=0,
\end{array} \qquad k=0,1,2,\ldots.\right.
\]

In a similar manner as in the Toda lattice case, one can also consider the dispersionless limit of the Pfaff
lattice equation. The dispersionless Pfaff (dPfaff) lattice equation is then given by
\begin{equation}\label{dDKPequation}
\left\{\begin{array}{lll}
\displaystyle{-4F_{13}+6F_{11}^2+3F_{22}=12e^{F_{00}}}\\[1.5ex]
\displaystyle{2F_{03}+F_{01}^3+6F_{01}F_{11}-3F_{01}F_{02}-6F_{12}=0}\,.
\end{array}\right.
\end{equation}
Here we have used the limits, $v^{\pm}=\exp(\log\tau_{n\pm 1}-\log\tau_n)\to \exp(\pm \hbar^{-1}F_0)$
 and $v^+v^-=\exp(\log\tau_{n+1}-2\log\tau_n+\log\tau_{n-1})\to  \exp(F_{00})$ as $\hbar\to 0$.

 In \cite{T:07}, Takasaki formulates the dispersionless DKP (Pfaff lattice) hierarchy by taking the dispersionless limit in the differential Fay identities of the Pfaff lattice hierarchy:
 \begin{equation}\label{dDKP-Fay}
 \left\{\begin{array}{llll}
  \displaystyle{e^{D(\lambda)D(\mu)F}\left(\lambda+\mu-(D(\lambda)+D(\mu))F_1-F_{01}\right)=
 \frac{\lambda^2e^{-D(\lambda)F_0}-\mu^2e^{-D(\mu)F_0}}{\lambda-\mu}},\\[2.5ex]
 \displaystyle{ e^{-D(\lambda)D(\mu)F}\left(
 1-\frac{(D(\lambda)-D(\mu))F_1}{\lambda-\mu}\right)=
 1-\frac{e^{F_{00}}}{\lambda^2\mu^2}e^{(D(\lambda)+D(\mu))F_0}}\,.
 \end{array}\right.
 \end{equation}
This equation is somewhat similar to the dToda hierarchy \eqref{dToda-Fay}, and in fact, this similarity
is one of our motivations to consider the universality between dPfaff and dToda hierarchies.


\subsection{Orthogonal and symplectic ensembles of random matrices}\label{OSensemble}
Two special classes of solutions of the Pfaff lattice hierarchy are given by the partition functions of the orthogonal and symplectic ensembles of random matrices (see for examples \cite{K:99, AvM:02, KP:07}).
We denote those partition functions by $Z^{(\beta)}$ with $\beta=1$ for the orthogonal ensembles
and $\beta=4$ for the symplectic ensembles. With this notation, the partition function (or $\tau$-function)
for the Toda hierarchy is denoted by $Z^{(2)}$, i.e. $\beta=2$.

\subsubsection{Orthogonal ensemble}
In this case, the partition function is taken to be
\[
 Z_{2n}^{(1)}(V_0(\lambda); \mathbf{t}) =\int_{\mathbb{R}}{d\lambda_1}\cdots\int_{\mathbb{R}}d\lambda_{2n}\,\prod_{i<j}
|\lambda_i- \lambda_j|\, \exp\left( - \sum_{k=1}^{2n} V_{\mathbf{t}}(\lambda_k) \right)\,,\]
where
$ V_{\mathbf{t}}(\lambda_k)=V_0(\lambda_k)-\sum_{j=1}^{\infty}\lambda_k^jt_j$ and $V_0(\lambda)$ is a polynomial of even degree.
This integral represents integration over the eigenvalues of the random matrices in the orthogonal ensemble, and as in the case of unitary ensemble, its Taylor expansion in $\mathbf{t}$ gives the moments of traces of powers of the matrix.

This partition function is also related to
 the combinatorial problem of counting ribbon graphs on a compact surface.
In this case we will consider M\"obius graphs, which are defined to be one-complexes embedded in an unoriented surface of Euler characteristic $\chi$ such that the
complement of the complex is a collection of sets homeomorphic to discs.  We may think of a M\"obius graph as a ribbon graph
where the ribbons are allowed to be glued together with a twist, i.e. reversing the local orientation of the vertices.  Our M\"obius graphs are
labeled in the same manner as the ribbon graphs.
We introduce M\"obius graphs here for completeness sake, however in this paper
we consider only those
M\"obius graphs of Euler characteristic $2$, such graphs are equivalent, up to
a choice of local orientations at the vertices,  to ribbon graphs
embedded in the sphere (see Lemma \ref{mobius-is-ribbon}),  so we will
not encumber our discussion here with examples of M\"obius graphs.
One may see \cite{MW-2003} for a discussion of M\"obius graphs with Euler
characteristics less than 2.

For a quadratic potential $V_0(\lambda)$, we have an asymptotic expansion for $\log(Z^{(1)}_{2n})$ of the
same basic character as that of the unitary ensembles:
\begin{Theorem}[Goulden-Harer-Jackson \cite{GHJ}] \label{ghj}
For the Gaussian orthogonal ensemble (GOE), we have
\[ \log\left[ Z_{2N}^{(1)}\left( \frac{N}{2} \lambda^2; 2N\mathbf{T} \right) \Big/
  Z_{2N}^{(1)}\left(\frac{N}{2} \lambda^2; \mathbf{0}\right)\right]
 = \sum_{\chi\leq 2} E^{(1)}_\chi(\mathbf{T}) (2N)^{\chi} \]
where
\[ E^{(1)}_\chi(\mathbf{T}) = \sum_{\mathbf{j}}
 \mathcal{K}_\chi(\mathbf{j}) \frac{\mathbf{T}^\mathbf{j}}{\mathbf{j}!}
 \,, \]
and $\mathcal{K}_\chi(\mathbf{j}) $ is the number of connected M\"obius
 graphs with $j_k$ labeled vertices of degree $k$ for $k=1,2,\ldots,$ on a compact surface of Euler characteristic $\chi$.
\end{Theorem}
We then take our $\tau$-function in the form with $T_0=\frac{n}{N}=n\hbar$,
\[
\tau_{n}(\hbar^{-1}\mathbf{T};\hbar)=Z_{2n}^{(1)}\left(\frac{N}{4} \lambda^2;  N \mathbf{T} \right)
= (2T_0)^{n^2+\frac{n}{2}}Z_{2n}^{(1)}\left( \frac{n}{2} \lambda^2; 2n\hat{\mathbf{T}}\right)   \,,
 \]
where $\mathbf{\hat{T}}$ is given by the Penner scaling, i.e.
$\hat{\mathbf{T}}=(\hat{T}_1,\hat{T}_2,\ldots)$ with $\hat{T}_j=(2T_0)^{\frac{j}{2}-1}T_j$.  This expression then has the asymptotic expansion
\begin{align*}
\log\left[ \tau_n(\hbar^{-1}\mathbf{T};\hbar) \right]
&= \sum_{\chi \leq 2} E^{(1)}_\chi( {\hat{\mathbf{T}}}) (2n)^{\chi} + C^{(1)}(T_0;N) \\
&= \sum_{\chi \leq 2} (2 T_0)^{\chi} E^{(1)}_\chi( {\hat{\mathbf{T}}})N^{\chi} +C^{(1)}(T_0;N)\,,
\end{align*}
where $C^{(1)}(T_0;N)=\log[Z_{2n}^{(1)}(\frac{N}{4}\lambda^2;\mathbf{0})]$.
Note that the form of this asymptotic expansion agrees (at leading order) with the assumption (\ref{asymptotic-assumption}).  For future reference, for this initial condition, the leading order of
the expansion gives a solution of the dPfaff hierarchy,
\begin{align*}
 F^{(1)}(T_0, \mathbf{T}) &=\lim_{N\to\infty}\,\frac{1}{N^2}\,\log\left[ Z_{2n}^{(1)}\left(\frac{N}{2} \lambda^2; N \mathbf{T}\right) \right]\\
 &= 4 T_0^2 E_2^{(1)}( \hat{\mathbf{T}})+C^{(1)}_0(T_0) \,,
\end{align*}
where $C_0^{(1)}(T_0)$ is the leading order of $C_0^{(1)}(T_0;N)$. In Appendix B, we explicitly calculate $C_0^{(1)}(T_0)$, and it gives
\[
C_0^{(1)}(T_0)=T_0^2\log(2T_0)-\frac{3}{2}T_0^2=\frac{1}{2}C_0^{(2)}(2 T_0)\,,
\]
where $C^{(2)}_0(T_0)=C_0(T_0)$ given in (\ref{C0Toda}).
The relation with $C^{(2)}(T_0)$ of the dToda hierarchy is a consequence of the universality, which will be discussed in Section
\ref{dToda-dPfaff}.

With $\hat{T}_j = (2T_0)^{\frac{j}{2}-1} T_j$, we have the expansion formula for the solution of the dPfaff hierarchy
corresponding to the Gaussian orthogonal ensemble,
 \begin{align}
  F^{(1)}(T_0, \mathbf{T}) &= (2T_0)^2 E_2^{(1)}(\hat{\mathbf{T}}) + C^{(1)}_0(T_0)\label{Ofree} \\
  &= (2 T_0)^2 \sum_{\mathbf{j}} \mathcal{K}_2(\mathbf{j}) (2 T_0)^{e-v} \frac{\mathbf{T}^\mathbf{j}}{\mathbf{j}!} + C^{(1)}_0(T_0)\nonumber \\
  &= \sum_{\mathbf{j}} \mathcal{K}_2(\mathbf{j}) (2T_0)^f \frac{\mathbf{T}^\mathbf{j}}{\mathbf{j}!} + C^{(1)}_0(T_0)\nonumber \,.
   \end{align}
This formula then gives the combinatorial meaning of the solution of the dPfaff hierarchy
(also notice the similarity with the formula (\ref{Ufree}) for the dToda hierarchy.

\subsubsection{Symplectic ensemble}
In this case,  the  partition function is given by
\[
Z_{n}^{(4)}(V_0(\lambda); \mathbf{t}) =\int_{\mathbb{R}}{d\lambda_1}\cdots\int_{\mathbb{R}}d\lambda_{n}\,\prod_{i<j} |\lambda_i - \lambda_j|^4\, \exp\left( - \sum_{k=1}^{n}2 V_{\mathbf{t}}(\lambda_k) \right)\,,\]
where
$ V_{\mathbf{t}}(\lambda_k)=V_0(\lambda_k)-\sum_{j=1}^{\infty}\lambda_k^jt_j$ and $V_0(\lambda)$ is a polynomial of even degree.
This function represents integration over the eigenvalues of the random matrices in the symplectic ensemble, and the moments of traces of powers of the matrix are obtained in its Taylor expansion in $\mathbf{t}$.

Then the corresponding combinatorial result is give by the following Theorem:
\begin{Theorem}[Mulase-Waldron \cite{MW-2003}] \label{mw}
For the Gaussian symplectic ensemble (GSE), we have
\[ \log\left[ Z_N^{(4)}\left( \frac{N}{2} \lambda^2; 2 N\mathbf{T}\right) \Big/ Z_N^{(4)}\left(\frac{N}{2}\lambda^2; 0\right)\right] = \sum_{\chi\leq 2} E^{(1)}_\chi(\mathbf{T}) (-2N)^{\chi}\,,
\]
where $E^{(1)}_\chi(\mathbf{T})$ is defined as in Theorem \ref{ghj}.
\end{Theorem}
This Theorem can also be found from Theorem 4.1 of \cite{BP:08}, in particular the $1/(4N)^{n-m}$ appearing in that theorem leads to the $2N \mathbf{T}$ in the above formula of Theorem \ref{mw}.
\begin{Remark} There is a duality in the Gaussian orthogonal and symplectic ensembles.
Theorems \ref{ghj} and \ref{mw} give the $N\mapsto -N$ duality between the
partition functions of the Gaussian orthogonal and symplectic ensembles of
random matrices.  In general the even terms of the asymptotic expansions agree,
while the odd terms have opposite signs.  See the discussions in \cite{MW-2003, BP:08}.
\end{Remark}

The $\tau$-function in this case is defined by
\[
\tau_n\left(\hbar^{-1}\mathbf{T};\hbar\right)=Z_n^{(4)}\left(\frac{N}{4} \lambda^2; N\mathbf{T}\right)
=(2T_0)^{n^2-\frac{n}{2}}Z_n^{(4)}\left(\frac{n}{2}\lambda^2;2n\hat{\mathbf{T}}\right)
\,,
\]
where $\hat{\mathbf{T}}$ is the Penner scaling defined in the same form as in the case of orthogonal ensemble, i.e. $\hat{\mathbf{T}}=(\hat T_1,\hat T_2,\ldots)$ with $\hat{T}_j = (2T_0)^{\frac{j}{2}-1} T_j$.  Then the free energy is given by
\begin{align*}
\log\left[\tau_n\left(\hbar^{-1}\mathbf{T};\hbar\right)\right]&=
\log\left[ Z_n^{(4)}\left( \frac{N}{4} \lambda^2 ; N \mathbf{T}\right) \Big/ Z_n^{(4)}\left( \frac{N}{4} \lambda^2; \mathbf{0}\right) \right] + C^{(4)}(T_0;N) \\
&=\log\left[ Z_n^{(4)}\left( \frac{n}{2} \lambda^2; 2 n \hat{\mathbf{T}}\right)\Big/ Z_n^{(4)}\left( \frac{n}{2} \lambda^2; 2n \hat{\mathbf{T}}\right) \right] + C^{(4)}(T_0;N) \\ &=
\sum_{\chi \leq 2} (-2T_0)^{\chi} E^{(1)}_\chi(\hat{\mathbf{T}}) N^{\chi} + C^{(4)}(T_0;N)\,,
\end{align*}
 where $C^{(4)}(T_0;N)=\log[Z_n^{(4)}(\frac{N}{4}\lambda^2;\mathbf{0})]$.  In the limit $N\to\infty$,
 the leading order of the free energy is given by
  \begin{align}\label{Sfree}
 F^{(4)}(T_0,\mathbf{T})&=
\lim_{N\to\infty}\frac{1}{N^2}\log\left[Z_n^{(4)}\left(\frac{N}{4}\lambda^2;N\mathbf{T}\right)\right]\\
&=(2T_0)^2E_2^{(1)}(\hat{\mathbf{T}})+C^{(4)}_0(T_0)\,,\nonumber
 \end{align}
 where $C^{(4)}_0(T_0)$ is the leading order of $C^{(4)}(T_0;N)$, and in Appendix B
 it is found to be the same as $C_0^{(1)}(T_0)$. This implies that with (\ref{Ofree})
 we have the duality of the free energies, i.e.
 \[
 F^{(1)}(T_0,\mathbf{T})=F^{(4)}(T_0,\mathbf{T}).
 \]
In Section \ref{dToda-dPfaff}, we further discuss the universality for those free energies
corresponding to the GUE, GOE and GSE random matrix theories.


\subsection{Outline of the paper}
The paper is organized as follows: In Section \ref{S:dToda}, we first show
that the dToda hierarchy
is determined by variables $F_{00}$ and $F_{01}$ (Proposition
\ref{F00F01}). Then we discuss some properties of $F_{mn}$ (Proposition
\ref{parity-proposition}). In particular, we show that $F_{1n}$ are related to
the Catalan numbers, which gives
a combinatorial meaning of the dToda hierarchy (Proposition \ref{CatalanToda}).
In Section \ref{combin-dToda}, we first derive the explicit formulae for the two
point functions $F_{nm}$
for the dToda hierarchy under the conditions of $F_{01}=0$ and $F_{11}=1$ (Theorem \ref{Fformula}), and
then discuss the combinatorial description of $F_{mn}$ as the enumeration of ribbon graphs on a sphere
(Proposition \ref{FmnGraph}). We also show that the combinatorial meaning of
$F_{mn}$ can be directly obtained from the dToda hierarchy (Theorem \ref{GFcomb}).
In Section \ref{dPfaffhierarchy}, we describe the dPfaff hierarchy. In
particular, we show that under the assumption of the special dependency of the free energy $F$ in the variables
$(T_0,\mathbf{T})$, the dPfaff hierarchy can be reduced to the
dToda hierarchy (Proposition \ref{SR-Todaconnection}). This corresponds to the
dispersionless limit
of the Pfaff lattice hierarchy restricted to symplectic matrices  \cite{KP:08}.
In Section \ref{dToda-dPfaff}, we first show the universality among the free energies for the
GUE, GOE and GSE in the leading order of large $N$ expansions (Proposition \ref{UOSrelations}).
Then we show in general that the solution of the dToda hierarchy also satisfies the dPfaff
hierarchy under the rescaling of the two point functions $F_{mn}$ (Theorem \ref{DKP-Toda}).
This means that the dPfaff hierarchy contains the dToda hierarchy, and
shows the universality of the dToda hierarchy at the leading orders among the unitary, orthogonal and symplectic ensembles in random
matrix theory. In Appendix A, we give a brief introduction of the dKP theory, and
derive the dKP hierarchy (\ref{Cauchykernel}) which plays an important role in this paper.
We also give in Appendix B explicit computations of the terms $C_0^{(\beta)}(T_0)$
for $\beta=1,2$ and $4$.

\section{The dispersionless Toda hierarchy}\label{S:dToda}
In this section, we first show that the dToda hierarchy \eqref{dToda-Fay} defines a deformation
of the algebraic
curve given by \eqref{todacurve}, and as a result we show that the two point functions $F_{nm}$ are
determined by
$F_{01}$ and $F_{00}$.   We also discuss a connection of the dToda hierarchy
with the dKP hierarchy, and define the Faber polynomials as the generators of the higher
members of the dKP hierarchy. These results will then be used in Section
\ref{combin-dToda} to derive explicit formulae of $F_{mn}$ solving the
two-vertex problem.

\subsection{The dToda curve and its deformation with the dToda hierarchy}\label{dTodaC}
Let us first define the function $p(\lambda)$,
\[
p(\lambda):=\lambda \exp\left(-D(\lambda)F_0\right).
\]
This corresponds to the equation $p(\lambda)=\frac{\partial S}{\partial T_0}$ with the $S$-function
given by (\ref{S}).
Then the dToda hierarchy is defined on the dToda curve \eqref{todacurve}:
\begin{Lemma}\label{tridiagonality}
The dToda hierarchy \eqref{dToda-Fay} gives a deformation of the algebraic curve of genus 0
given by
\[
\lambda=p+F_{01}+\frac{e^{F_{00}}}{p}\,.
\]
\end{Lemma}
\begin{Proof}
Eliminating $e^{D(\lambda)D(\mu)F}$ from two equations in \eqref{dToda-Fay}, we obtain
\[
1 = \left(\frac{p(\lambda)-p(\mu)}{\lambda-\mu}\right)\left(1-\frac{e^{F_{00}}}{p(\lambda)p(\mu)}\right).
\]
 We then separate the variables,
\[
\lambda -p(\lambda)-\frac{e^{F_{00}}}{p(\lambda)}=\mu-p(\mu)-\frac{e^{F_{00}}}{p(\mu)}.
\]
With the asymptotic condition, $\lambda-p(\lambda)\to F_{01}+O(\lambda^{-1})$, we obtain the curve.
\end{Proof}

With this Lemma \ref{tridiagonality}, we define in this paper
the dispersionless Toda (dToda) hierarchy in the form,
\begin{equation}\label{dTodahierarchy}\left\{\begin{array}{llllll}
\displaystyle{1-\frac{e^{F_{00}}}{p(\lambda)p(\mu)}}=e^{-D(\lambda)D(\mu)F}\\[1.0ex]
\displaystyle{\lambda=p(\lambda)+F_{01}+\frac{e^{F_{00}}}{p(\lambda)} \quad {\rm with}\quad p(\lambda)=\lambda e^{-D(\lambda)F_0}}\,.
\end{array}\right.
\end{equation}
In particular one should note through the paper that the dToda curve (the second equation)
on $(\lambda,p)$-plane plays a central role for the combinatorial
description of the dToda hierarchy.

We now note that the first equation in \eqref{dTodahierarchy} implies:
\begin{Lemma}\label{f0k}
Each $F_{mn}$ with $m,n\ge 1$ can be determined by the set $\{F_{0k}: 0\le k\le m+n\}$.
\end{Lemma}
\begin{Proof}  The first  equation with $p(\lambda)=\lambda \exp(-D(\lambda)F_0)$ gives
\[
D(\lambda)D(\mu)F=\sum_{m\ge 1}\sum_{n\ge 1}\frac{F_{mn}}{mn}\frac{1}{\lambda^m\mu^n}=-\log\left(1-\frac{e^{F_{00}}}{\lambda\mu}e^{(D(\lambda)+D(\mu))F_0}\right)\,.
\]
The right hand side depends only on $F_{0k}$ for $k\ge 0$.
\end{Proof}
Also, setting $\lambda=\mu$ and expanding for large $\lambda$, the first equation gives
\begin{align*}
1-\frac{e^{F_{00}}}{p(\lambda)^2}&= 1-\frac{e^{F_{00}}}{\lambda^2}\sum_{n=0}^{\infty}
\frac{1}{\lambda^n}h_n(2\hat{\mathbf{F}}_0) \\[1.0ex]
&=e^{-D(\lambda)^2F}=\sum_{n=0}^{\infty}\frac{1}{n}h_n(-\hat{\mathbf{Z}})\,,
\end{align*}
where $\hat{\mathbf{F}}_0=(\hat{F}_{01},\hat{F}_{02},\ldots)$ with $\hat{F}_{0k}=\frac{1}{k}F_{0k}$, and
$\hat{\mathbf{Z}}=(0,\hat{Z}_2,\hat{Z}_3,\ldots)$ with $\hat Z_n$ for $n\ge 2$ defined by
\[
\hat Z_n=\sum_{k+l=n}\frac{F_{kl}}{kl} \qquad(k,l\ge 1).
\]
Comparing the coefficients of the power $\lambda^{-n}$, we obtain
\begin{equation}\label{dToda1}
h_{n+2}(-\hat{\mathbf{Z}})+e^{F_{00}}h_{n}(2\hat{\mathbf{F}}_0)=0\qquad n\ge 0.
\end{equation}
This set of equations is called the dipersionless Hirota equation for the dToda hierarchy,
which gives a half set of the dispersionless Hirota equations (see for example \cite{CT:06},
and the other half is given by \eqref{dToda2} below).
The first equation with $n=0$ then gives the dispersionless Toda equation,
\[
F_{11}=e^{F_{00}}\,.
\]
For $n=1$, we obtain
\[
F_{12}=2e^{F_{00}}F_{01}=2F_{01}F_{11}.
\]
Thus $F_{11}$ and $F_{12}$ are determined by only $F_{00}$ and $F_{01}$. In fact, we have
 from the second equation (dToda curve) that
\begin{Proposition}\label{F00F01}
The two point functions $F_{nm}$ of the dToda hierarchy is determined by $F_{00}$ and $F_{01}$ only.
\end{Proposition}
\begin{Proof}
The curve gives
\[
1= e^{-D(\lambda)F_0}+\frac{F_{01}}{\lambda}+\frac{e^{F_{00}}}{\lambda^2}e^{D(\lambda)F_0}.
\]
Expanding this with large $\lambda$, we have
\[
h_n(-\hat{\mathbf{F}}_0)
+e^{F_{00}}h_{n-2}(\hat{\mathbf{F}}_0)=0\quad n\ge 2.
\]
 Then one can see that $F_{0n}$ is determined by the previous $F_{0k}$ for $0\le k<n$.
 For examples, we have for $n=2$ and $3$,
\begin{align*}
F_{02}&=F_{01}^2+2e^{F_{00}}=F_{01}^2+2F_{11}\,,\\
F_{03}&=\frac{3}{2}F_{01}F_{02}-\frac{1}{2}F_{01}^3+3F_{01}e^{F_{00}}=F_{01}^3+6F_{01}F_{11}\,.
\end{align*}
Then from Lemma \ref{f0k}, we conclude that all $F_{mn}$ for $m,n\ge 1$ are determined by
$F_{00}$ and $F_{01}$ (or $F_{01}$ and $F_{11}=\exp F_{00}$).
\end{Proof}

  From the dToda curve, $p(\lambda)$ can be explicitly calculated as follows:
\begin{Proposition}\label{CatalanToda}
The inverse of the dToda curve for the case $p(\lambda)\to\lambda$ as $\lambda\to\infty$ is give by
\[
p(\lambda)=\lambda-F_{01}-\sum_{n=0}^{\infty}\frac{F_{1,n+1}}{(n+1)\lambda^{n+1}}\,,
\]
where
\[
\frac{F_{1,n+1}}{n+1}=\sum_{k=0}^{[\frac{n}{2}]}C_k\binom{n}{2k}F_{01}^{n-2k}F_{11}^{k+1}\,,
\qquad {\rm with}\quad C_k=\frac{1}{k+1}\binom{2k}{k}\,.
\]
Note here that $F_{11}=e^{F_{00}}$ and $C_k$ is the $k$-th Catalan number
with $C_0=1$.
\end{Proposition}
\begin{Proof}
 From the curve, first we have
\[
p^2-(\lambda-F_{01})p+F_{11}=0.
\]
Solving this for $p$ with the asymptotic condition $p(\lambda)\to \lambda$ as
$\lambda\to\infty$,
we have
\begin{align*}
p&= \frac{1}{2}\left(\lambda-F_{01}+\sqrt{(\lambda-F_{01})^2-4F_{11}}\,\right) \\[0.5ex]
&=\lambda-F_{01}-\frac{1}{2}\left(\lambda-F_{01}-\sqrt{(\lambda-F_{01})^2-4F_{11}}\,\right)\,.
\end{align*}
Noting the last term can be expressed as
\[
\frac{F_{11}}{\lambda-F_{01}}\frac{1-\sqrt{1-4x}}{2x}=\frac{F_{11}}{\lambda-F_{01}}\sum_{k=0}^{\infty} C_kx^k\qquad {\rm with}\quad x=\frac{F_{11}}{(\lambda-F_{01})^2}.
\]
Thus we have
\[
\sum_{n=0}^{\infty}\frac{F_{1,n+1}}{(n+1)\lambda^{n+1}}=\sum_{k=0}^{\infty}C_k\frac{F_{11}^{k+1}}{(\lambda-F_{01})^{2k+1}}.
\]
Then the coefficients $P_n$ can be found by
\[
\frac{F_{1,n+1}}{n+1}= \sum_{k=0}^{\infty}C_kF_{11}^{k+1}\oint_{\lambda=\infty}\frac{d\lambda}{2\pi i}\,\frac{(\lambda+F_{01})^{n}}{\lambda^{2k+1}}.
\]
Using the expansion $(\lambda+F_{01})^n=\sum_{j=0}^{n}\binom{n}{j}\lambda^jF_{01}^{n-j}$,
we obtain the formula.
\end{Proof}
The appearance of the Catalan numbers indicates some connections of combinatorial problems
to the dToda hierarchy. This is one of the main motivations of the present study.

\subsection{The dToda hierarchy implies the dKP hierarchy}\label{dToda-dKP}

Here we consider the second equation of \eqref{dToda-Fay}, that is, in terms of $p(\lambda)$,
\begin{equation}\label{Cauchykernel}
\frac{p(\lambda)-p(\mu)}{\lambda-\mu}=e^{D(\lambda)D(\mu)F}\,.
\end{equation}
This equation is also known as the dKP hierarchy, and
$p(\lambda)$ in the dKP hierarchy is given by $p(\lambda)=\lambda-D(\lambda)F_1$ (see Appendix A, and also for examples   \cite{CK:95, CT:06}).
For the dToda hierarchy, $p(\lambda)=\lambda\exp(-D(\lambda)F_0)$ can also be expressed in a similar form.
This is obtained from an interesting connection with the Faber polynomials:
\begin{Definition}\label{faberFG}
Let $\lambda$ be a Laurent series in $p$ given by
\[
\lambda=p+u_1+\sum_{k=1}^{\infty}\frac{u_{k+1}}{p^k}.
\]
Then the Faber polynomials $\Phi_n(p)$ of degree $n$ in $p$ are defined by
\[
\Phi_n(p)=\left[\lambda(p)^n\right]_{\ge 0} \quad n=0,1,2,\ldots,
\]
where $[\lambda(p)^n]_{\ge 0}$ represents the polynomial part of $\lambda(p)^n$ in $p$.
The polynomial $\Phi_n(p(\lambda))$ has the Laurent series in $\lambda$,
\[
\Phi_n(p(\lambda))=\lambda^n-\sum_{m=1}^{\infty}\frac{1}{\lambda^m}Q_{nm}\,,
\]
where $Q_{nm}$ are the functions of $\mathbf{u}=(u_1,u_2,\ldots)$. Those $Q_{nm}$ are related to the
Grunsky coefficients $c_{nm}=\frac{1}{n}Q_{nm}$. Then \eqref{Cauchykernel} gives the generating
function of the Faber polynomials with $c_{nm}=\frac{1}{nm}F_{nm}$
 (see Remark \ref{FG} below for a further discussion). In the dKP theory, the Faber polynomials give
 the generators of the flows in the hierarchy (see Appendix A).
\end{Definition}
 Then with (\ref{Cauchykernel}), we have
\begin{Proposition}\label{faber}
The Faber polynomials can be expressed as
\[
\Phi_n(p(\lambda))=\lambda^n-D(\lambda)F_n \qquad n\ge 1\,,
\]
and $\Phi_0(p)=1$.
\end{Proposition}
\begin{Proof}
We first note that
\[
\log\left(\frac{p(\lambda)-p(\mu)}{\lambda-\mu}\right)=D(\lambda)D(\mu)F.
\]
Taking the derivative with respect to $\lambda$, we have
\[
\frac{1}{p(\lambda)-p(\mu)}\frac{\partial p(\lambda)}{\partial \lambda}=\frac{1}{\lambda-\mu}-\sum_{n=1}^{\infty}\frac{1}{\lambda^{n+1}}D(\mu)F_n\,.
\]
Then for $n\ge 1$, we have
\begin{align*}
&\oint_{\lambda=\infty}\frac{d\lambda}{2\pi i}\frac{\lambda^n}{p(\lambda)-p(\mu)}\frac{\partial p(\lambda)}{\partial\lambda}=\mu^n-D(\mu)F_n \\[0.5ex]
&=\oint_{p=\infty}\frac{dp}{2\pi i}\frac{\lambda(p)^n}{p-p(\mu)}= \left[\lambda(p(\mu))^n\right]_{\ge 0}.
\end{align*}
This completes a proof.
\end{Proof}

We now derive the other set of the dispersionless Hirota equations:
Taking the limit $\lambda\to\mu$ of the equation in (\ref{Cauchykernel}), we have
\[
\frac{\partial p(\lambda)}{\partial \lambda}=e^{D(\lambda)^2F}\,.
\]
With $p(\lambda)=\lambda\exp(-D(\lambda)F_0)$, we have
\begin{align*}
-\lambda\frac{\partial}{\partial \lambda}D(\lambda)F_0&=e^{D(\lambda)^2F+D(\lambda)F_0}-1\\
&=\exp\left(\sum_{n=1}^{\infty}\frac{1}{\lambda^n}Z_n\right)-1=\sum_{n=1}^{\infty}\frac{1}{\lambda^n}h_n(\mathbf{Z})\,,
\end{align*}
where $\mathbf{Z}=(Z_1,Z_2,\ldots)$ with $Z_k$ defined by
\[
Z_1=F_{01},\quad Z_n=\frac{F_{0n}}{n}+\sum_{k+l=n}\frac{F_{kl}}{kl}\,.
\]
This then gives the other half set of the dispersionless Hirota equation for the dToda hierarchy
(the first half set is given by (\ref{dToda1})),
\begin{equation}\label{dToda2}
F_{0n}=h_n(\mathbf{Z})\quad n\ge 2,
\end{equation}
which gives the dispersionless limit of the hierarchy \eqref{dHT2}.
The first equation of (\ref{dToda2}) with $n=2$  gives
\[
F_{02}=F_{01}^2+2F_{11}\,.
\]

With Proposition \ref{faber} with the case $n=1$, i.e. $\Phi_1(p)=p+F_{01}$, the function $p(\lambda)$ can also be expressed by
\[
p(\lambda)=\lambda-F_{01}-D(\lambda)F_1.
\]
Then note that that  the dHirota equation of the dKP hierarchy is derived from the same equation $\partial p/\partial \lambda=e^{D(\lambda)^2F}$ with $p(\lambda)=\lambda-F_{01}-D(\lambda)F$,
\[
F_{1n}=h_{n+1}(\hat{\mathbf{Z}})\qquad n\ge 3,
\]
which is the dispersionless limit of \eqref{KPhierarchy}.
Here recall that $\hat{\mathbf{Z}}=(0,\hat{Z}_2,\hat{Z}_3,\ldots)$ with $\hat Z_k$ is given by $\hat Z_k=Z_k-F_{0k}/k$.
The first equation with $n=3$ gives the dKP equation,
\begin{equation}\label{dKPequation}
-4F_{13}+6F_{11}^2+3F_{22}=0\,.
\end{equation}

We also mention the positivity of $F_{mn}$:
\begin{Proposition} \label{parity-proposition}
The coefficients $F_{mn}$ satisfy the following properties:
\begin{itemize}
\item[(a)] If $F_{1m}\ge 0$ for all $m\ge 1$, then $F_{mn}\ge 0$ for all $m,n\ge 1$.
\item[(b)] If $F_{1,2k}=0$ for all $k\ge 1$ and all others $F_{1,m}\ge0$, then $F_{mn}=0$ for all $m,n$ with $m+n=$odd.
\end{itemize}
\end{Proposition}
\begin{Proof}
With $p(\lambda)=\lambda-F_{01}-D(\lambda)F_1$, we have
\begin{align*}
e^{D(\lambda)D(\mu)F}&=\frac{p(\lambda)-p(\mu)}{\lambda-\mu}=1-\frac{(D(\lambda)-D(\mu))F}{\lambda-\mu}\\
&=1-\frac{1}{\lambda-\mu}\sum_{m=1}^{\infty}\frac{F_{1m}}{m}\left(\frac{1}{\lambda^m}-\frac{1}{\mu^m}\right)\\
&=1+\frac{1}{\lambda\mu}\sum_{m=1}^{\infty}\frac{F_{1m}}{m}\left(\sum_{k=1}^m\frac{1}{\lambda^{m-k}\mu^{k-1}}\right).
\end{align*}
Then taking the logarithm and the differentiation with respect to $\lambda$ of this equation, we have
\[
\sum_m\sum_n\frac{F_{mn}}{n}\frac{1}{\lambda^{m+1}\mu^n}=e^{-D(\lambda)D(\mu)F}\sum_{m}
\frac{F_{1m}}{m}\left(\sum_{k=1}^m\frac{m-k+1}{\lambda^{m-k+2}\mu^{k-1}}\right).
\]
This shows the property (a), i.e.  $\forall F_{1m}\ge 0$ implies $\forall F_{mn}\ge 0$.

The second property (b) can be obtained from $\partial p(\lambda)/\partial\lambda=e^{D(\lambda)^2F}$.
This equation gives
\begin{align*}
D(\lambda)^2F=\sum_{n=2}^{\infty}\frac{1}{\lambda^n}\left(\sum_{k+l=n}\frac{F_{kl}}{kl}\right)&=\log\left(
1-\frac{\partial}{\partial \lambda}D(\lambda)F_1\right)\\
&=\log\left(1+\sum_{k=1}^{\infty}\frac{F_{1k}}{\lambda^{k+1}}\right),
\end{align*}
which shows that if all $F_{1,2k}=0$, then
\[
\sum_{n+m={\rm odd}}\frac{F_{mn}}{mn}=0.
\]
Then from (a), we conclude $F_{mn}=0$ if $m+n=$odd.
\end{Proof}
\begin{Remark}\label{FG}
The Faber polynomials $\Phi_n(p)$ appear in a classical complex function theory:
The Bieberbach conjecture sates that if $f(z)$ is a univalent function on $|z|<1$
with the expansion $f(z)=z+a_2z^2+a_3z^3+\cdots$, then the coefficients $a_n$ satisfy
$|a_n|\le n$. Let $g(z)$ be the function defined by
\[
g(z)=\frac{1}{f\left(\frac{1}{z}\right)},
\]
which is univalent for $|z|>1$.
Then the Faber polynomials $\Phi_n(p)$ appear in a expansion,
\[
\log\left(\frac{g(z)-p}{z}\right)=-\sum_{n=1}^{\infty}\frac{\Phi_n(p)}{nz^n}\,.
\]
With the expansion of $\Phi_n(g(z))$ for large $z$,
\begin{equation}\label{cnm}
\Phi_n(g(z))=z^n+n\sum_{m=1}^{\infty}\frac{c_{nm}}{z^m}\,.
\end{equation}
Grunsky defined the Grunsky coefficients $c_{nm}$. Then Grunsky showed that his coefficients should satisfy a sequence of inequalities in order that
$g$ is univalent on $|z|>1$.
In particular, an explicit formula for the Gransky coefficient was found in the paper by Schur in terms of
the coefficients $a_n$ of $f$ \cite{S:45}.  In the section \ref{explicitFnm}, we find an explicit formula of $F_{nm}$ for the case where the cefficients are given by the Catalan numbers and discuss their combinatorial
meaning (our derivation is independent from the calculation used in \cite{S:45}).

  From (\ref{Cauchykernel}) and (\ref{cnm}), one can see that the Grunsky coefficients $c_{nm}$ are related to our two point functions, that is, $
c_{nm}=\frac{1}{nm}F_{nm}$.
The dToda hierarchy then defines an {\it integrable} deformation with an infinite number of parameters $\mathbf{T}=(T_1,T_2,\ldots)$
for the univalent function $f(z)$, so that the Grunsky coefficients are given by the second derivatives of
a function $F$, the free energy, i.e.
\begin{equation}\label{Grunsky}
c_{nm}=\frac{1}{nm}\frac{\partial^2F}{\partial T_n\partial T_m}\,.
\end{equation}
Thus a most important consequence of the dToda hierarchy is to show the existence of a unique function $F$, the free energy.
\end{Remark}


\section{Combintorial Results for the dToda hierarchy} \label{combin-dToda}

The appearance of the Catalan numbers in the dToda hierarchy (see Proposition \ref{CatalanToda}) directly demonstrates a combinatorial meaning of
the functions $F_{nm}$. In particular, if $F_{01}(1; \mathbf{0})=0$ and $F_{11}(1; \mathbf{0})=1$, then we have, from Proposition \ref{CatalanToda},
\[
p(\lambda)=\lambda-\sum_{k=0}^{\infty}\frac{C_k}{\lambda^{2k+1}},
\]
which gives
\begin{equation}\label{F1m}
F_{1,2k+1}=(2k+1) C_k\,.
\end{equation}
In general, the two point functions $F_{nm}$ have the following combinatorial interpretation:
One sees from the connection to random matrices discussed in section \ref{UnitaryToda} that
for the particular choice of $F(T_0,\mathbf{T}) = T_0^2 e_0(\mathbf{\hat{T}}) +C_0(T_0)$ the two point function with $n,m\ge 1$,
\[ F_{nm}(T_0 = 1; \mathbf{T}=\mathbf{0}) = \frac{\partial^2 F}{\partial T_n \partial T_m} \bigg|_{T_0=1; \mathbf{\hat{T}}=\mathbf{0} }
 = \frac{\partial^2 }{\partial T_n \partial T_m} T_0^2 e_0(\mathbf{\hat{T}})
 \bigg|_{T_0=1; \mathbf{\hat{T}}=\mathbf{0}}\,,
\]
is the number of genus 0 connected ribbon graphs with a vertex of degree $n$
and a vertex of degree $m$ (see Theorem \ref{BIZ-thm}).  Here the degree of the vertex represents the
number of edges (ribbons) attached to the vertex.
For example, the $F_{1,2k+1}$ gives the number of connected graphs with a
vertex with a single edge
and a vertex with $2k+1$ edges on a sphere (recall that the number of ways to
connect $2k$ edges
of a single vertex is given by the Catalan number $C_k$).
In this section we compute a remarkable closed formula for the $F_{n,m}(1; \mathbf{0})$ in the case with $F_{01}(1;\mathbf{0}) = 0$ and $F_{11}(1;\mathbf{0}) = 1$.  Also recall that differentiation with respect to $T_0$ has the meaning of counting ribbon graphs with a marked face.

\subsection{Explicit formulae for $F_{nm}(1;\mathbf{0})$}\label{explicitFnm}
Let us first recall Proposition \ref{parity-proposition} (b):  if $F_{1,2k}=0, \forall k>0$ and $F_{1,m}\ge0,
\forall m>0$, we have
\[
F_{nm}=0 \qquad {\rm if}\quad n+m={\rm odd}\,.
\]
Also from Proosition \ref{CatalanToda}, if $F_{10}=0$ and $F_{11}=1$, we have $F_{1,2k}=0$ and
$F_{1,2k+1}=(2k+1)C_k$. So we calculate only $F_{nm}$ with $m+n=$even and positive.
Then we state the following Theorem which gives an explicit formula of $F_{nm}$ for those cases of $m$ and $n$. The method to find the formula is based on the dispersionless Toda hierarchy.
\begin{Theorem}\label{Fformula}
The two point functions $F_{nm}=\frac{\partial^2}{\partial T_n\partial T_m}F$ for the dToda hierarchy (\ref{dTodahierarchy}) with
$F_{01}=0$ and $F_{11}=1$ (i.e. $F_{00}=0$) are given by
\[\left\{\begin{array}{rlll}
F_{0,2k}&=&(k+1)C_k,\qquad k=1,2,\ldots, \\[2.0ex]
\displaystyle{F_{2j+1,2k+1}}&=&{(2j+1)(2k+1)\frac{(j+1)(k+1)}{j+k+1}C_jC_k},\qquad j,k=0,1,2,\ldots,\\[2.0ex]
\displaystyle{F_{2j,2k}}&=&{jk\frac{(j+1)(k+1)}{j+k}C_jC_k},\qquad j,k=1,2,\ldots,\\[2.0ex]
F_{nm}&=&0, \qquad {\rm otherwise}\,,
\end{array}\right.
\]
where $C_j$ is the $j$-th Catalan number.
\end{Theorem}
\begin{Proof}
We first derive the formula $F_{0,2k}$ for $k\ne 0$:
We start with the definition of $p$,
\[ p(\lambda) = \lambda e^{-D(\lambda) F_0}\,. \]
which implies
\[  D(\lambda) F_0 =\sum_{k=1}^{\infty}\frac{F_{0,2k}}{2k}\frac{1}{\lambda^{2k}}= \log\left( \frac{\lambda}{p(\lambda)} \right)\,. \]
Therefore we may compute $F_{0, 2k}$ by the contour integral
\[ \frac{1}{2k} F_{0, 2k} =  \oint _{\lambda=\infty}\frac{d\lambda}{2\pi i}\,\lambda^{2k-1} \log\left( \frac{\lambda}{p(\lambda)}\right)\,.
\]
Using the dToda curve,
\[ \lambda(p) = p + \frac{1}{p}=p\left(1+\frac{1}{p^2}\right) \,,\]
we write the integral with respect to $p$, and
then integrate by parts,
\begin{align*}
\frac{1}{2k} F_{0, 2k} &=
 \oint _{p=\infty} \frac{dp}{2\pi i}\,\frac{1}{2k}\left(\frac{d}{dp}\lambda(p)^{2k}\right)\, \log\left( 1 + \frac{1}{p^2}\right)  \\
&= \oint_{p=\infty}\frac{dp}{2\pi i}\,\frac{1}{k}\, p^{2k-3} \left(1+\frac{1}{p^2}\right)^{2k-1} = \frac{1}{k} \binom{2k-1}{k-1} \,.
\end{align*}
In the last step, we have used the expansion, $(1+\frac{1}{p^2})^{2k-1}=\sum_{j=0}^{2k-1}\binom{2k-1}{j}p^{-2j}$.
Thus we obtain
\begin{equation*}\label{F0}
F_{0,2k}=(k+1) C_k\,.
\end{equation*}

We now derive the formula $F_{nm}$ for $nm\ne 0$: We use the Faber polynomials in Proposition
\ref{faber}, i.e.
\[
\Phi_n(p)=[\lambda(p)^n]_{\ge 0}=\lambda^n-D(\lambda)F_n\,.
\]
This implies
\[
\frac{F_{mn}}{m}=-\oint_{\lambda=\infty}\frac{d\lambda}{2\pi i}\,\lambda^{m-1}\Phi_n(p(\lambda))\,.
\]
Now changing variables from $\lambda$ to $p$, we have
\begin{align*}
{F_{mn}} &= -\oint_{p=\infty}\frac{dp}{2\pi i}\,\left(\frac{d}{dp}\lambda(p)^m\right)
\Phi_n(p) \\
&=-\oint_{p=\infty}\frac{dp}{2\pi i}\,\left(p+\frac{1}{p}\right)^n\,\frac{d}{dp}\left[\left(p+\frac{1}{p}\right)^m\right]_{\le 0}\,.
\end{align*}
Here we have noticed that $\Phi_n(p)$ can be replaced by $\lambda(p)^n$ with
$d([\lambda(p)^m]_{<0}])$ for $d(\lambda(p)^m)$. The notation $[f(p)]_{<0}$ implies
the part of the Laurent series $f(p)$ with the inverse powers of $p$, i.e. $[f(p)]_{<0}=f(p)-[f(p)]_{\ge0}$.
Thus we have
\begin{equation}\label{fmnint}
{F_{mn}}=\sum_{k=0}^{\left[\frac{m}{2}\right]}(m-2k)\binom{m}{k}\oint_{p=\infty}
\frac{dp}{2\pi i}\,\left(p+\frac{1}{p}\right)^np^{2k-m-1}\,.
\end{equation}
Then from the expansion $(p+\frac{1}{p})^n=\sum_{l=0}^n\binom{n}{l}p^{n-2l}$, we have the residue
integral
\[
\oint_{p=\infty}\frac{dp}{2\pi i}\,p^{n-m+2(k-l)-1}\,,
\]
which implies $n-m+2(k-l)=0$ in the non-zero case. Therefore, $F_{mn}\ne 0$ only if $n-m$ is even (recall Proposition \ref{parity-proposition}). So we have two cases: (a) $m=2j+1, n=2k+1$, and
(b) $m=2j, n=2k$. With the residue calculation for the case (a), we obtain
\[
{F_{2j+1,2k+1}}=\sum_{l=0}^{{\rm min}(j,k)}(2l+1)\binom{2j+1}{j-l}\binom{2k+1}{k-l}\,.
\]

We consider the case (a) (the case (b) can be shown in the similar way):
Let us assume $k\ge j$ (because of the symmetry in $j,k$, this case is enough to show).
Also confirm that the case of $j=0$ gives (\ref{F1m}).
Then from (\ref{fmnint}), we have
\[
F_{2j+1,2k+1}=\sum_{l=0}^j(2j-2l+1)\binom{2j+1}{l}\binom{2k+1}{k-j+l}.
\]
We note the following identity which is a key equation in the proof,
\begin{equation}\label{key}
(2j-2l+1)\binom{2j+1}{l}=(2j+1)\left[\binom{2j}{l}-\binom{2j}{l-1}\right]\,.
\end{equation}
Also we write
\[
\frac{1}{2k+1}\binom{2k+1}{k-j+l}=\frac{k(k-1)\cdots(k-j+l+1)}{(k+2)(k+3)\cdots(k+j-l+1)}\,C_k\,.
\]
Then we have the following formula with the common denominator $(k+2)(k+3)\cdots(k+j+1)$,
\begin{align*}
\frac{F_{2j+1,2k+1}}{(2j+1)(2k+1)}&=\sum_{l=0}^j\left[\binom{2j}{l}-\binom{2j}{l-1}\right]
\frac{k(k-1)\cdots(k-j+l+1)}{(k+2)(k+3)\cdots(k+j-l+1)}\,C_k \\
&=\frac{C_k}{(k+2)(k+3)\cdots(k+j+1)}G(j,k)\,,
\end{align*}
where $G(j,k)$ is the polynomial of degree $j$ in $k$
\begin{equation}
 \label{G1-virgil}
 G(j, k) := \sum_{l=0}^j \left[ \binom{2j}{l} - \binom{2j}{l-1} \right] \prod_{\alpha=0}^{j-l-1} (k-\alpha) \times \prod_{\alpha=0}^{l-1} (k+j+1-\alpha) \,.
\end{equation}
Here the products $\displaystyle{\prod_{\alpha=0}^{-1}}(\cdots)$ are understood to give 1 when $l=0$ and $l=j$.
  We then claim that
the polynomial $G(j,k)$ can be expressed by the simple form,
\begin{align}\label{G}
G(j,k) &= \binom{2j}{j}(k+1)(k+2)\cdots (k+j)\\
&=(j+1)C_j(k+1)(k+2)\cdots (k+j)\,.\nonumber
\end{align}
Here note that the coefficient $\binom{2j}{j}$ is just $\sum_{l=0}^j\left[\binom{2j}{l}-\binom{2j}{l-1}\right]$.
This then gives the desired formula, that is,
this equation for $G(j,k)$ leads to the formula $F_{2j+1,2k+1}$.

 To prove the claim, we write (\ref{G1-virgil}) in the sum of
the terms with the coefficients $\binom{2j}{l}$ (like a telescoping sum):
\begin{equation}\label{G1-B}
G(j, k) = - 2 \sum_{l=0}^{j-1} (j-l) \binom{2j}{l} \prod_{\alpha=0}^{j-l-2} (k-\alpha) \times \prod_{\alpha=0}^{l-1} (k+j+1-\alpha)
+ \binom{2j}{j} (k+2) (k+3) \dots (k+j+1) \,,
\end{equation}
where the last term in (\ref{G1-B}) is the remainder from the telescoping.
We then observe that this reduces the degree of the polynomials of both sides of (\ref{G}), that is (\ref{G}) becomes the equation
\begin{equation}\label{G1}
 2\sum_{l=0}^{j-1}(j-l)\binom{2j}{l} \prod_{\alpha=0}^{j-l-2} (k-\alpha) \times \prod_{\alpha=0}^{l-1} (k+j+1-\alpha)
=j\binom{2j}{j}(k+2)(k+3)\cdots (k+j) \,,
\end{equation}
where the right hand side of this equation is the difference between the  remainder in (\ref{G1-B}) and the right hand side of (\ref{G}).
We then note that the coefficients in the sum can be written in the form similar to (\ref{key}),
\[
(j-l)\binom{2j}{l}=j\left[\binom{2j-1}{l}-\binom{2j-1}{l-1}\right]\,.
\]
Here we have used the relations (the second one is the Pascal rule),
\[
l\binom{2j}{l}=2j\binom{2j-1}{l-1}\qquad{\rm and}\qquad  \binom{2j}{l}=\binom{2j-1}{l}+\binom{2j-1}{l-1}\,.
\]
Then we can write (\ref{G1}) in a similar form to (\ref{G}),
\[
\sum_{l=0}^{j-1}\left[\binom{2j-1}{l}-\binom{2j-1}{l-1}\right]
\prod_{\alpha=0}^{j-l-2} (k-\alpha) \times \prod_{\alpha=0}^{l-1} (k+j+1-\alpha)
=\binom{2j-1}{j-1}(k+2)(k+3)\cdots(k+j)\,.
\]
Again we write the sum in the terms with the coefficients $\binom{2j-1}{l}$ and move the remainder to the right hand side to arrive at the equation
\[
\sum_{l=0}^{j-2}(2j-2l-1)\binom{2j-1}{l}
\prod_{\alpha=0}^{j-l-3} (k-\alpha) \times \prod_{\alpha=0}^{l-1} (k+j+1-\alpha)
=\binom{2j-1}{j-1}(k+3)(k+4)\cdots (k+j).
\]
Both sides are now polynomials of degree $j-2$ in $k$. In the similar way with the identity,
\[
(2j-2l-1)\binom{2j-1}{l}=(2j-1)\left[\binom{2j-2}{l}-\binom{2j-2}{l-1}\right]\,,
\]
the above equation takes the similar form as before, i.e.
\[
\sum_{l=0}^{j-2}\left[\binom{2j-2}{l}-\binom{2j-2}{l-1}\right]
\prod_{\alpha=0}^{j-l-3} (k-\alpha) \times \prod_{\alpha=0}^{l-1} (k+j+1-\alpha)
=\binom{2j-2}{j-2}(k+3)(k+4)\cdots (k+j).
\]
Then we can further reduce the degree of the polynomials in the same way. Continuing this process,
we complete the proof.
\end{Proof}


\subsection{Graph enumeration on the sphere}

We recall that $F_{nm}$ has a combinatorial meaning of enumeration of
connected (ribbon) graphs on a sphere.
 The presentation of such combinatoric problems in terms of the function $F(T_0; \mathbf{T})$ arising from the Gaussian unitary ensemble example is not a new idea.  However viewing the genus zero solutions as arising purely from the dToda hierarchy is a new approach to solving the problem, although a similar idea was used in the case of a single even time for the computations of \cite{EMP-2008}.
 Surprisingly we have found a closed form for the numbers $F_{nm}$ with the conditions $F_{01}=0$
 and $F_{11}=1$ for general $n$ and $m$.

The solution of the two-vertex problem given by Theorem \ref{Fformula}
is new, previously
 the problem has been solved only in the so called ``dipole" case when the
 vertices are of the same degree (i.e. $F_{nn}$) \cite{J:1993, KL:93}.  These authors consider the general genus problem of two vertices of the same degree on a genus $g$ surface.
 Similar formulae were also found for other low genus problems in \cite{EMP-2008, BIZ}, but again the vertices must all have the same degree.

 Let us start by giving a combinatoric argument for the $F_{0, 2k}$  formula of
Theorem \ref{Fformula}, that is, $F_{0,2k}$  gives the number of genus 0 ribbon graphs
with a vertex of degree $2k$ and a marked face:
As shown above counting graphs with marked faces in the case of fixed genus is
equivalent to just counting the number of graphs, the number of faces being
given by the vertex structure.  So the problem is then to count the number of
genus 0 ribbon graphs with a single vertex of degree $2k$.
It is easy to see that this is equivalent to counting the number of ways to
make pair-wise connections
with $2k$ points on a circle  without crossing and the number of sections
divided by the chords giving
the connections of the vertices. The first number is well-knwon and is given
by the $k$-th Catalan number $C_k$.
Then the sphere is divided up to $k+1$ sections. So we have
$(k+1)C_k$ which is $F_{0,2k}$ (see Theorem \ref{Fformula}).

Now we consider the two vertices problem, and show that $F_{nm}$ satisfy the following
recursion relations
as a direct consequence of the enumeration of connected ribbon graphs
appearing in the problem:
\begin{Proposition}\label{FmnGraph}
The two point functions $F_{m n}$ with $nm\ne 0$ for the two-vertex problem satisfy the recursion relations,
\[\left\{\begin{array}{lll}
F_{2j+1, 2k+1} =2 \sum_{l=1}^{j} C_{j-l} F_{2l- 1, 2k+1} + (2k+1)C_{j+k}
\qquad j,k=0,1,2,\ldots\\
F_{2j, 2k} =  2 \sum_{l=2}^{j} C_{j-l} F_{2l - 2, 2k} + 2k C_{j+k-1} \qquad j,k=1,2,\ldots
\end{array}\right.
\]
\end{Proposition}
\begin{Proof}
\begin{figure}
\begin{center}
\includegraphics[width=9cm]{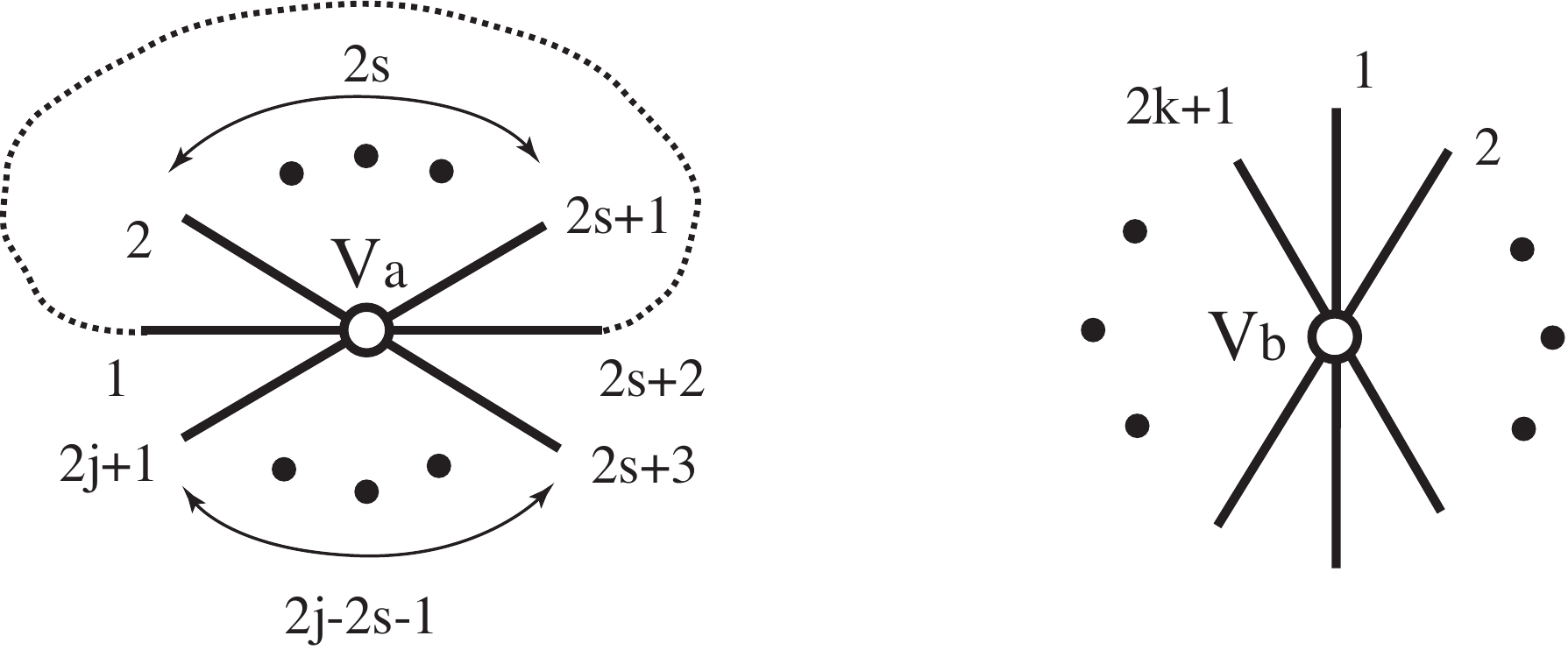}
\caption{Two-vertex problem with degrees $2j+1$ and $2k+1$.
The gluing of the edge $1$ of $V_a$ to the edge $2s+2$ of the same vertex divides into the two
sets of edges; the set of $2s$ edges marked by $\{2,3,\ldots,2s-1\}$ is of the one-vertex problem, since
those edges can not cross the glued ribbon between the edges $1$ and $2s+2$ on a sphere.
The gluing of the $2j-2s-1$ edges of the vertex $V_a$ to $2k+1$ edges of $V_b$
gives the two-vertex problem.\label{isolate}
}
\end{center}
\end{figure}
We consider the case of $F_{2j+1,2k+1}$ (the case of $F_{2j,2k}$ can be shown in the same way).
We will use $V_a$ and $V_b$ to refer to the vertex of degree $2j+1$ and $2k+1$ respectively.
We start with the graph, in which a specific edge, say 1, of $V_a$ is glued
to an edge of the vertex $V_b$. The number of ways to accomplish this gluing is $2k+1$ (the number of edges at $V_b$).
Having glued the edge 1 of $V_a$ with an edge of $V_b$, we have the situation of one-vertex problem
of degree $2j+2k$. So the number of connected graphs for each situation is given by the Catalan number $C_{j+k}$. Thus the total number of the edge 1 gluing is $(2k+1)C_{j+k}$, which is the
last term in the equation.
Now we calculate the other cases where the edge 1 is glued to another edge of $V_a$.
Let  the edge 1 be glued with the edge $2s+2$ ($s=0,1,\ldots,j-1$).
This divides the edges of $V_a$ into two sets one of which is cut off the rest of the graph (see Figure \ref{isolate}).
This group of edges is marked by $\{2, 3, \ldots, 2s+a\}$, and these $2s$ edges of $V_a$ must be glued together within this group.  This is the one-vertex problem of degree $2s$, and the number of
gluings is $C_{s}$.  Now we are left with the two vertices problem with the vertex of degree $2j-2s-1$
for $V_a$ and the vertex of degree $2k+1$. The number of connected graphs of this problem is,
of course, $F_{2j-2s-1,2k+1}$.  This situation with one-vertex problem of degree $2s$ and
two-vertex problem of degrees $2j-2s-1$ and $2k+1$ appears also the case where the edge 1
is glued with the edge $2j-2s+1$. Here the group of $2s$ edges in the one-vertex problem is $\{2j-2s+2,
2j-2s+3,\ldots,2j+1\}$.
So we have the total number of connected graphs is
\[
F_{2j+1,2k+1}= 2\sum_{s=0}^{j-1}C_sF_{2j-2s-1,2k+1}+(2k+1)C_{j+k}\,.
\]
Then changing the variables in the sum with $s=j-l$, we have the equation.
\end{Proof}
We also note that these relations in Proposition \ref{FmnGraph} define a recursion relation of the Faber polynomials on the dToda curve:
We consider only the first equation (the second equation can be shown in the same manner).
Multiplying $\frac{1}{(2k+1)\lambda^{2k+1}}$ to the first equation in Proposition \ref{FmnGraph} and summing up for $k=0,1,2\ldots$, we obtain, after rearranging the terms,
\begin{equation}\label{Frecursion}
\Phi_{2j+1}-2\sum_{l=1}^jC_{j-l}\Phi_{2l-1}=\lambda^{2j}\Phi_1-\sum_{l=1}^jC_{j-l}\lambda^{2l-1}\,.
\end{equation}
(Recall that $\Phi_n=\lambda^n-\sum_{m=1}^{\infty}\frac{1}{m\lambda^m}F_{nm}$ and
$\Phi_1=p=\lambda-\sum_{k=0}^{\infty}\frac{1}{\lambda^{2k+1}}C_k$.) This can be considered as
a recursion relation of the Faber polynomials defined on the dToda curve $\lambda=p+\frac{1}{p}$, that is,
with \eqref{Frecursion}, one can define $\Phi_{2j+1}$ from $\Phi_1=p$. Conversely, starting from this recursion, one can derive the first equation of Proposition \ref{FmnGraph}.  This implies that we have an alternative proof of this
Proposition using only the Faber polynomials, that is, the explicit form of two point functions $F_{mn}$ obtained in Theorem \ref{Fformula} satisfy the recursion relations of Proposition \ref{FmnGraph}.

Let us now show that the relation \eqref{Frecursion} can be directly derived from the dToda curve. We first note that the left hand side is a polynomial in $p$.  It is indeed easy to see that the projection of the right hand side on the polynomial part gives the polynomial in the right hand side. This means that  the non-polynomial part in the right hand side must vanish, when $(\lambda, p)$ is on the dToda curve, i.e. $\lambda=p+\frac{1}{p}$. So we want to prove
\begin{equation}\label{NonP}
\left[\lambda(p)^{2j}p\right]_-=\left[\sum_{l=1}^jC_{j-l}\lambda(p)^{2l-1}\right]_-\,,
\end{equation}
where $[f(p)]_-$ is the non-polynomial part of the Laurent series $f(p)$. With the dToda curve, this leads to
\[
\sum_{k=0}^{j-1}\binom{2j}{k}p^{-2(j-k)+1}=\sum_{l=0}^{j-1}\sum_{k=0}^{j-l-1}C_l\binom{2(j-l)-1}{k}p^{-2(j-k)+2l+1}\,.
\]
Then from the coefficients of the power $p^{-2(j-k)+1}$ and using the Pascal relation, $\binom{2j}{k}=\binom{2j-1}{k}+\binom{2j-1}{k-1}$, we have the following Lemma
which proves equation \eqref{NonP}:
\begin{Lemma} For $j>k$, we have
\[
\binom{2j-1}{k-1}=\sum_{l=1}^kC_l\binom{2(j-l)-1}{k-l}\,.
\]
\end{Lemma}
\begin{Proof}
Let us first note that we may write the second binomial coefficient as
\[ \binom{2 (j-l) -1}{(j-l) - (j-k)} =\binom{2n-1}{n-i}\,,
\]
where we denoted $n=j-l$ and $i=j-k$. We also note
\begin{equation}\label{binom}
\binom{2n-1}{n-i} = \frac{(n+i)}{2i} \frac{i}{n} \binom{2n}{n-i}.
\end{equation}

Now let $z(t)$ be a generating function of the Catalan numbers,
\[
z=\sum_{j=0}^{\infty}C_jt^j\,,
\]
which satisfies $1=z-tz^2$ with $z=1$ at $t=0$.
Then one can show that
\[
(z-1)^i=\sum_{n=i}^{\infty}\frac{i}{n}\binom{2n}{n-i}t^n\,.
\]
This result can be found by computing the contour integral, using $t=\frac{z-1}{z^2}$ and $dt=\frac{2-z}{z^3}dz$,
\begin{align*}
 \oint_{t=0}\frac{dt}{2\pi i}\, \frac{1}{t^{n+1}}(z-1)^i  &=  \oint_{z=1} \frac{dz}{2\pi i}\frac{ z^{2n-1} }{(z-1)^{n-i+1} } (2-z)  \\
&= \oint_{u=0} \frac{du}{2\pi i}\,\frac{ (u+1)^{2n-1} }{ u^{n-i+1} } (1 - u) \\
&= \binom{2n-1}{n-i} - \binom{2n-1}{n-i-1} = \frac{i}{n} \binom{2n}{n-i} \,.
\end{align*}
Then from \eqref{binom}, we see that
\begin{align*}
 \binom{2n-1}{n-i} &= \oint_{t=0}\frac{dt}{2\pi i} \,\frac{1}{t^{n+i}} \frac{1}{2i} \frac{d}{dt} \left[  t^i (z-1)^i \right] \\
 &= \oint_{t=0}\frac{dt}{2\pi i} \,\frac{1}{2t^{n+1}}  (z-1)^{i-1} \left( z-1 + t z' \right) \\
 &=  \oint_{t=0}\frac{dt}{2\pi i} \,\frac{1}{2t^{n+1}}  (z-1)^i \left(1 + \frac{z}{2-z} \right)\\
 &=  \oint_{t=0}\frac{dt}{2\pi i} \,\frac{1}{t^{n+1}}  \frac{(z-1)^i }{2-z} \,.
 \end{align*}
 This implies that
 \[
 \frac{(z-1)^i}{2-z}=\sum_{n=i}^{\infty}\binom{2n-1}{n-i}t^n\,.
 \]
Also with $C_0=1$, we have $z-1=\sum_{l=1}^{\infty}C_lt^l$, so that we see by setting $i=j-k$,
\begin{align*}
(z-1)\frac{(z-1)^i}{2-z}&=\sum_{l=1}^{\infty}\sum_{n=j-k}^{\infty}C_l\binom{2n-1}{n-(j-k)}t^{n+l}\\
&=\sum_{m=j-k+1}^{\infty}\sum_{l=1}^{m-j+k}C_l\binom{2(m-l)-1}{(m-l)-(j-k)}t^m\,.
\end{align*}
 Then the sum in the Lemma is given by the coefficient of $t^j$, i.e.
 \begin{align*}
  \sum_{l=1}^k C_l \binom{2(j-l)-1}{k-l}  &= \oint_{t=0} \frac{dt}{2\pi i}\,\frac{1}{t^{j+1}}(z-1)\frac{ (z-1)^{j-k} }{ 2-z } \\
 &= \oint_{z=1} \frac{dz}{2\pi i}\, \frac{ z^{2j-1}}{(z-1)^k} = \oint_{u=0}  \frac{du}{2\pi i}\, \frac{ (u+1)^{2j-1}}{u^k} \\
 &= \binom{2j-1}{k-1} \,.
 \end{align*}
  \end{Proof}
 This Lemma proves the relation \eqref{NonP} which leads to the recursion relation of the Faber polynomials \eqref{Frecursion} defined on the dToda curve. The point here is that the combnatorial meaning of the two point functions $F_{mn}$ given in Proposition \ref{FmnGraph} can be
directly obtained from the structure of the Faber polynomials defined on the dToda curve (i.e.
independent from the random matrix description of the two point functions). We summarize this as the following Theorem:
\begin{Theorem}\label{GFcomb}
The Grunsky coefficients $c_{nm}=\frac{1}{nm}F_{nm}$ of the Faber polynomials defined on the dToda curve $\lambda=p+\frac{1}{p}$ satisfy the recursion relations given in Proposition \ref{FmnGraph}, that is, they have the combinatorial meaning of counting the connected ribbon graphs for the two-vertex
problem on a sphere.
\end{Theorem}


\section{The dispersionless Pfaff hierarchy}\label{dPfaffhierarchy}
In this section, we consider some basic properties of the dPfaff hierarchy.
In particular, we show that the dPfaff hierarchy can be reduced to the dToda hierarchy
by choosing a particular section of the solution space which depends only on the even times
of the hierarchy (i.e. the symplectic Pfaff lattice hierarchy studied in \cite{KP:08}).
We also discuss a rescaling of the variables based on the structure of the $\tau$-functions
of the Toda and Pfaff lattice hierarchies.  That structure is that one is given by determinants and
other by Pfaffians.

\subsection{The dispersionless limits of the symplectic Pfaff lattice hierarchy}
Let us first note that the dPfaff hierarchy \eqref{dDKP-Fay} reduces to the dToda hierarchy
as a special choice of the variables.
\begin{Proposition}\label{SR-Todaconnection}
Let $\alpha=\lambda^2$ and $\beta=\mu^2$. Suppose $F_{2m+1}=\frac{\partial F}{\partial t_{2m+1}}=0$
for all $m\ge 0$, and let $s_n=2t_{2n}$ for $n\ge 1$ and $s_0=t_0$. Then the dPfaff hierarchy
(\ref{dDKP-Fay})
is reduced to the dToda hierarchy \eqref{dToda-Fay},
\[\left\{\begin{array}{llll}
\displaystyle{e^{\tilde D(\alpha)\tilde D(\beta)\tilde F}=\frac{\alpha e^{-\tilde D(\alpha)\tilde F_0}-\beta e^{-\tilde D(\beta)\tilde F_0}}{\alpha-\beta}}\\[2.0ex]
\displaystyle{e^{-\tilde D(\alpha)\tilde D(\beta)\tilde F}=1-\frac{e^{\tilde F_{00}}}{\alpha\beta}e^{(\tilde D(\alpha)+\tilde D(\beta))\tilde F_0}}\,,
\end{array}\right.
\]
where $\tilde F(s_0,s_1,\ldots)=F(t_0,t_1,\ldots)$ and $\tilde D(\alpha)=\sum_{k=1}^{\infty}\frac{1}{k\alpha^k}\frac{\partial}{\partial s_k}$.
\end{Proposition}
\begin{Proof}
 The proof is based on the following observation:  With $\alpha=\lambda^2$ and $F_{2m+1}=0$  $(\forall m\ge0)$,
\begin{align*}
D(\lambda)F&=\sum_{m=1}^{\infty}\frac{1}{2m\lambda^{2m}}\frac{\partial F}{\partial t_{2m}}\\
&=\sum_{m=1}^{\infty}\frac{1}{m\alpha^m}\frac{\partial \tilde F}{\partial s_m}=\tilde D(\alpha)\tilde F\,.
\end{align*}
It is then immediate to see that \eqref{dDKP-Fay} with $F_1=0$ gives \eqref{dToda-Fay}.
\end{Proof}
This can be considered as the dispersionless limit of the Pfaff lattice restricted to symplectic matrices (a continuous version of the SR algorithm \cite{KP:08}).
The dPfaff  equation is then given by
\[
F_{22}=4e^{F_{00}} \quad {\rm which ~is~}\quad \tilde F_{11}=e^{\tilde F_{00}}\,.
\]
The second  equation in the dPfaff (dDKP) equation (\ref{dDKPequation}) becomes trivial.

\subsection{The dPfaff curve and the rescaled dPfaff hierarchy}
We consider the following scalings:
\begin{equation}\label{scalings}
T_0\to 2T_0, \quad \mathbf{T}\to 2\mathbf{T},\quad {\rm and}\quad F\to \frac{1}{2}F\,.
\end{equation}
The scaling in the time variables may be considered as the size of the matrix models, i.e.
the matrices we consider for the orthogonal ensembles are twice the size of matrices for the unitary ensembles.  The scaling in the free energy function
is due to the structure of the $\tau$-functions, i.e. the determinant for the Toda lattice and
the Pfaffian for the Pfaff lattices (a Pfaffian is a square root of a determinant).
We will show in Section \ref{dToda-dPfaff} that those scalings appear in the particular
choices of our $\tau$-functions corresponding to the Gaussian ensembles discussed
in the Introduction.

With the scalings (\ref{scalings}), $F_{mn}$ is scaled by $2F_{mn}$ for all $m,n\ge 0$, and we redefine the dPfaff hierarchy \eqref{dDKP-Fay} in the follwoing form similar to the
dToda hierarchy,
\begin{equation}\label{dDKPrescaled}\left\{
\begin{array}{llll}
&\displaystyle{1-\frac{e^{2F_{00}}}{p(\lambda)^2p(\mu)^2}}=e^{-2D(\lambda)D(\mu)F}\left(
\frac{q(\lambda)-q(\mu)}{\lambda-\mu}\right) \\[1.5ex]
&\displaystyle{q(\lambda)^2=p(\lambda)^2+G_2+\frac{e^{2F_{00}}}{p(\lambda)^2}}\,,
\end{array}\right.
\end{equation}
with $p(\lambda)=\lambda e^{-D(\lambda)F_0}, q(\lambda)=\lambda-F_{01}-2D(\lambda)F_1$, and
\[
G_2=F_{02}-F_{01}^2-4F_{11}.
\]
We call the algebraic equation of $(p,q)$ in \eqref{dDKPrescaled} the dPfaff curve.
The first equation is the scaled version of the second one in \eqref{dDKP-Fay}, and the
first one in \eqref{dDKP-Fay} is derived from this definition \eqref{dDKPrescaled}:
\begin{Lemma}\label{2nddDKP} From \eqref{dDKPrescaled}, we have
\[
\frac{p(\lambda)^2-p(\mu)^2}{\lambda-\mu}=e^{2D(\lambda)D(\mu)F}\left(q(\lambda)+q(\mu)\right),
\]
which is the first equation in \eqref{dDKP-Fay} with the scaling $F_{mn}\to2F_{mn}$.
\end{Lemma}
\begin{Proof}
  From the dPfaff curve,  we have
\begin{align*}
q(\lambda)^2-q(\mu)^2
&=\left(p(\lambda)^2-p(\mu)^2\right)\left(1-\frac{e^{2F_{00}}}{p(\lambda)^2p(\mu)^2}\right)\\[0.5ex]
&=\left(p(\lambda)^2-p(\mu)^2\right)e^{-2D(\lambda)D(\mu)F}\left(\frac{q(\lambda)-q(\mu)}{\lambda-\mu}
\right),
\end{align*}
which leads to the equation in the Lemma.
\end{Proof}
This shows the equivalence between \eqref{dDKP-Fay} and \eqref{dDKPrescaled}.

Let us now discuss the functional relations among $F_{mn}$, i.e. the dispersionless Hirota equations
for the dPfaff hierarchy:  We first note
\begin{Lemma}
Each $F_{mn}$ for $m,n\ge 1$ can be determined by the set $\{F_{0k}:0\le k\le m+n\}\cup\{F_{11}\}$.
\end{Lemma}
\begin{Proof}
  From the dDKP curve, $q^2=p^2+G_2+e^{2F_{00}}/p^2$, we can see that each $F_{1,m}$ for
  $m\ge 2$ can be determined by $\{F_{0k}:0\le k\le m+1\}\cup\{F_{11}\}$.
  Then from the first equation,  we have
 \[
 2D(\lambda)D(\mu)F=-\log\left(1-\frac{e^{2F_{00}}}{p(\lambda)^2p(\mu)^2}\right)
 +\log\left(1-2\frac{(D(\lambda)-D(\mu))F_1}{\lambda-\mu}\right).
 \]
 This implies the assertion.
\end{Proof}

In the similar way as in the dToda hierarchy, the dispersionless Hirota equations can be derived
as follows:  From the dPfaff curve, we have for $n\ge 1$
\begin{equation}\label{dHirota-Pfaff}
\begin{array}{lll}
h_{n+2}(-2\hat{\mathbf{F}}_0)+e^{2F_{00}}h_{n-2}(2\hat{\mathbf{F}}_0)\\[1.5ex]
\displaystyle{=-\frac{4}{n+1}F_{1,n+1}+\frac{4}{n}F_{01}F_{1n}+4\sum_{j+k=n}\frac{F_{1j}}{j}\frac{F_{1k}}{k}\,,}
\end{array}
\end{equation}
where $\hat{\mathbf{F}}_0=(\hat{F}_{01},\hat{F}_{02},\ldots)$ with $\hat{F}_{0j}=\frac{F_{0j}}{j}$.
This gives the first half of the dispersionless Hirota equations for the dPfaff hierarchy, which shows
that  $F_{1n}$ are determined by $\{F_{0m}:0\le m\le n+1\}$. The first
nontrivial equation with $n=1$ is given by
\[
F_{12}=\frac{1}{3}F_{03}-F_{01}F_{02}+2F_{01}F_{11}+\frac{2}{3}F_{01}^3\,.
\]
This is the second equation of the dPfaff lattice equation (\ref{dDKPequation}) with the scaling
$F_{nm}\to 2F_{nm}$.

The second half of the dispersionless Hirota equations is given by
taking $\lambda=\mu$ in the first equation, i.e.
\[
1-\frac{e^{2F_{00}}}{p(\lambda)^4}=e^{-2D(\lambda)^2F}\frac{dq}{d\lambda}.
\]
Expanding for large $\lambda$, this gives
\[\begin{array}{lllll}
\displaystyle{1-\frac{e^{2F_{00}}}{\lambda^4}\sum_{n=0}^{\infty}\frac{1}{\lambda^n}
h_n(4\hat{\mathbf{F}}_0)} \\[2.0ex]
\displaystyle{=\left(\sum_{n=0}^{\infty}\frac{1}{\lambda^{n}}h_n(-2\hat{\mathbf{Z}}))\right)\left(1+2\sum_{n=1}^{\infty}\frac{1}{\lambda^{n+1}}F_{1n}\right)}\,.
\end{array}
\]
Then we have for $n\ge4$
\begin{equation}\label{dDKPhierarchy1}
h_n(-2\hat{\mathbf{Z}})+ \sum_{k+j=n-1}2F_{1k}h_j(-2\hat{\mathbf{Z}})
=-e^{2F_{00}}h_{n-4}(4\hat{\mathbf{F}}_0),
\end{equation}
where $\hat{\mathbf{Z}}=(0,\hat{Z}_2,\hat{Z}_3,\ldots)$ with $\hat{Z}_n=\sum_{j+k=n}\frac{F_{jk}}{jk}$.
The first equation with $n=4$ gives
\[
-\frac{2}{3}F_{13}+2F_{11}^2+\frac{1}{2}F_{22}=
e^{2F_{00}}.
\]
This is the first equation of the dispersionless limit of the Pfaff lattice equation (\ref{dDKPequation}) with the scaling $F_{mn}\to 2F_{mn}$.

\begin{Remark}
The dPfaff curve in \eqref{dDKPrescaled} has the structure on $(P,Q)\in\mathbb{C}^2$ with
$P=p^2$ and $Q=q^2$,
\[
Q=P+G_2+\frac{e^{2F_{00}}}{P}\,,
\]
which is the same as the dToda curve representing a sphere. The curve has also only two parameters
of the deformation. However, in terms of the coordinates $(\lambda,p)\in \mathbb{C}^2$, the curve
has an infinite number of parameters $F_{0k}$, unlike the dToda case.
In $(p,q)$-coordinates, the dPfaff curve has an asymptotic property
\[
q=p+\sum_{m=1}^{\infty}\frac{G_{2m}}{2mp^{2m-1}}\,,
\]
where $G_{2m}$ are determined by $G_2$ and $F_{00}$.
Then one can consider the Faber polynomials associated with the function $q(p)$.
However it seems that there is no role for those polynomials in the dPfaff theory.
\end{Remark}

\section{Universality of the dToda hierarchy}\label{dToda-dPfaff}

In this section we first show that the combinatoric descriptions for
the functions $e_0(\mathbf{T})$ and $E^{(1)}_2(\mathbf{T})$ in the leading order terms
of the free energies (see Theorems  \ref{BIZ-thm} and \ref{ghj}) are intimately related.
The fundamental fact underlying this is that the Euler characteristic two
M\"{o}bius graphs are equivalent to genus 0 ribbon graphs together with some
data specifying local orientations at the vertices.  The relationship between
the leading orders of the asymptotic expansions of the free energies $\log(Z^{(\beta)}_N)$ in the GUE ($\beta=2$),
GOE ($\beta=1$) and GSE ($\beta=4$) random matrix settings, implies that for the special initial condition
corresponding to these examples the solutions of the dToda and dPfaff
hierarchies agree (after a rescaling with proper choice of initial conditions).  We will then show that this is merely
an example of a more general fact: the solution of the dToda hierarchy is in
fact always a solution of the rescaled dPfaff hierarchy, or in other words the
dToda hierarchy implies the dPfaff hierarchy.
However the two hierarchies are not equivalent, the family of solutions to the
dPfaff hierarchy being much larger (see  \cite{KP:07})

\subsection{Universality in random matrix theory}

In broad terms universality in random matrix theory states that the limiting
behavior of eigenvalues of the unitary ensembles of random matrices is
independent of the external potential $V_0(\lambda)$ in certain scaling regimes.
There
is also a universality for the limiting behavior of the eigenvalues of the
random matrices across the orthogonal, unitary, and symplectic ensembles in
some regions and choices of scalings, see \cite{DG:2007A, DG:2007B} and \cite{S:07}.
In the Gaussian case, when $V_0(\lambda)$ is
quadratic, we will illustrate universality explicitly by giving a simple
relationship between the leading order expressions of the
free energy $F^{(\beta)}(T_0,\mathbf{T})$ for the GUE and GOE random matrix models
(recall that $F^{(1)}=F^{(4)}$, that is, the GOE and GSE cases give the same leading order).
A related result is that there is a universality between integrals over the unitary and orthogonal groups (not the same expressions as our $Z_n^{(2)}$ and $Z_n^{(1)}$) after appropriate scalings \cite{CGM-S:08}.

Let us first note the following Lemma relating the leading order combinatoric interpretation given by Theorems \ref{BIZ-thm} and \ref{ghj}.
\begin{Lemma} \label{mobius-is-ribbon}
M\"obius graphs of Euler characteristic 2 are in 1-1 correspondence with pairs made up of a ribbon graph of genus 0 together with one of the $2^{v-1} $ possible configurations for the local orientations of the vertices, where $v$ is the number of vertices of the graph
\end{Lemma}
\begin{Proof}
It is easy to see that a ribbon graph of genus 0 together with a choice of local orientations of the $v$ vertices gives a M\"obius graph of Euler characteristic 2.  For a M\"obius graph to have Euler characteristic 2 there must be a choice of the local orientations of the vertices so that the ribbons are untwisted (as the graph lays flat on a sphere), thus giving the pair of a ribbon graph of genus 0 and a choice of local orientations.
\end{Proof}
  From this lemma,  we see that (for leading order only) there is the following relation between the coefficients
$\kappa_0(\mathbf{j})$ and ${\mathcal{K}}_2(\mathbf{j})$ of the expansions
of the free energies (\ref{Ufree}) and (\ref{Ofree}),
\[ \mathcal{K}_2(\mathbf{j}) = 2^{v-1} \kappa_0(\mathbf{j}), \]
where $v = \sum_ij_i$ is the number of vertices.
We now have:
\begin{Proposition}\label{UOSrelations}
The free energies $F^{(\beta)}(T_0,\mathbf{T})$ have the following relations,
\[
F^{(1)}(T_0,\mathbf{T})=
F^{(4)}(T_0,\mathbf{T})=\frac{1}{2}F^{(2)}(2T_0,2\mathbf{T}).
\]
\end{Proposition}
\begin{Proof}
The duality between $F^{(1)}$ and $F^{(4)}$ has already been shown in section \ref{OSensemble}.
The relation between $F^{(1)}$ and $F^{(2)}$ is a direct consequence of the relation between $\kappa_0(\mathbf{j})$ and $\mathcal{K}_2(\mathbf{j})$ above. Namely from (\ref{Ofree}), we have
\[
\begin{array}{llll}
F^{(1)}(T_0,\mathbf{T})&=&\displaystyle{\sum_{\mathbf{j}}\mathcal{K}_2(\mathbf{j})(2T_0)^f\frac{\mathbf{T}^{\mathbf{j}}}{\mathbf{j}!}+C^{(1)}_0(T_0)}\\[2.0ex]
&= &\displaystyle{\sum_{\mathbf{j}}\kappa_0(\mathbf{j})2^{v-1}(2T_0)^f\frac{\mathbf{T}^{\mathbf{j}}}{\mathbf{j}!}+C_0^{(1)}(T_0)}\\[2.0ex]
&=& \displaystyle{\frac{1}{2}\,\sum_{\mathbf{j}}\kappa_0(\mathbf{j})(2T_0)^f\frac{(2\mathbf{T})^{\mathbf{j}}}{\mathbf{j}!}
+\frac{1}{2}C_0^{(2)}(2T_0)}\\[2.0ex]
&=&\displaystyle{\frac{1}{2}F^{(2)}(2T_0,2\mathbf{T})}\,.
\end{array}
\]
\end{Proof}
This Proposition implies that the free energy of the dToda hierarchy gives an universality
over those Gaussian ensembles in the leading order with the scalings given in (\ref{scalings}).
Also notice that the scalings in (\ref{scalings}) explicitly appear in this result.

\begin{Remark}
It is interesting to note that the scalings  in Proposition \ref{UOSrelations}, which give universality between the dToda and dPfaff
hierarchies, agree with the scalings found in \cite{CGM-S:08} giving the universality at leading order between the matrix integrals over the orthogonal and unitary groups (see Theorems 7.1 and 9.1 of \cite{CGM-S:08}).  In our case the partition functions are derived from integrals over the spaces of symmetric and hermitian matrices.
\end{Remark}

\subsection{The dToda hierarchy implies the {d}Pfaff hierarchy}

We now show that a solution of the dToda hierarchy \eqref{dTodahierarchy} gives a solution of the dPfaff hierarchy
\eqref{dDKPrescaled}. This extends the previous results for special choices of the free energies associated with the Gaussian ensembles, and it implies that the universality holds
for the general choice of the potential $V_0(\lambda)$ or the measure in the random matrix theory.
\begin{Theorem}\label{DKP-Toda}
The free energy $F$ of the dToda hierarchy \eqref{dTodahierarchy} also satisfies the dPfaff
hierarchy \eqref{dDKPrescaled}.
\end{Theorem}
\begin{Proof}
We first calculate the square of the first equation of the dToda hierarchy,
\begin{align*}
1 &= e^{2D(\lambda)D(\mu)F}\left(1-2\frac{e^{F_{00}}}{p(\lambda)p(\mu)}+\frac{e^{2F_{00}}}{p(\lambda)^2p(\mu)^2}\right)\\
&=e^{2D(\lambda)D(\mu)F}\left[2\left(1-\frac{e^{F_{00}}}{p(\lambda)p(\mu)}\right)-
\left(1-\frac{e^{2F_{00}}}{p(\lambda)^2p(\mu)^2}\right)\right]\\
&=2 e^{D(\lambda)D(\mu)F}-e^{2D(\lambda)D(\mu)F}\left(1-\frac{e^{2F_{00}}}{p(\lambda)^2p(\mu)^2}\right)\,.
\end{align*}
Then from Lemma \ref{Cauchykernel} of the kernel formula of $p(\lambda)$, we have
\[
e^{2D(\lambda)D(\mu)F}\left(1-\frac{e^{2F_{00}}}{p(\lambda)^2p(\mu)^2}\right)=2\frac{p(\lambda)-p(\mu)}{\lambda-\mu}-1.
\]
Now recall that $p(\lambda)$ for the dToda hierarchy is expressed by $p(\lambda)=\lambda-F_{01}-D(\lambda)F_1$. Then noting that $2p(\lambda)=q(\lambda)+\lambda-F_{01}$, we obtain the
first equation of the dPfaff hierarchy \eqref{dDKPrescaled}.

Now we show that the dPfaff curve $q^2=p^2+G_2+{e^{2F_{00}}}{p^{-2}}$ is reduced to
the dToda curve.  If $F_{mn}$ are the solutions of the dToda hierarchy, then the $G$ in the curve is given by
\[
G_2=-4F_{11}-F_{01}^2+F_{02}=-2F_{11}=-2e^{F_{00}}.
\]
Here we have used $F_{02}=F_{01}^2+2F_{11}$ and $F_{11}=e^{F_{00}}$ (see section \ref{dTodaC}).
This implies
\[
q^2=p^2-2e^{F_{00}}+\frac{e^{2F_{00}}}{p^2}=\left(p-\frac{e^{F_{00}}}{p}\right)^2.
\]
With $q=\lambda-F_{01}-2D(\lambda)F_1=2p-\lambda+F_{01}$, this equation leads to the dToda curve,
which is equivalent to the choice, $q=p-\frac{e^{F_{00}}}{p}$.
\end{Proof}
Finally, we remark that this universality can be seen in the dPfaff equation with
 the rescaling and the proper normalization:
The dPfaff equation is given by
\[
-\frac{2}{3}F_{13}+2F_{11}^2+\frac{2}{4}F_{22}=e^{2F_{00}}\,.
\]
which can be split to
\[
2\left(-\frac{1}{3}F_{13}+\frac{1}{2}F_{11}^2+\frac{1}{4}F_{22}\right)=-F_{11}^2+e^{2F_{00}}.
\]
Then if $F$ satisfies the dToda hierarchy, then both sides of this equation are zero.
Notice that the left hand side is the dKP equation and the other side is the dToda equation.
This splitting can be also seen in the Pfaff lattice equation (\ref{DKPequation}), i.e.
the first equation in \eqref{DKPequation} can be written in the form,
\[
\frac{\partial}{\partial t_1}\left(-4\frac{\partial u}{\partial t_3}+\frac{\partial^3u}{\partial t_1^3}+6u\frac{\partial u}{\partial t_1}\right)+3\frac{\partial^2u}{\partial t_2^2}=-3\frac{\partial^2}{\partial t_1^2}(
u^2-4v^+v^-)\,.
\]
Then one has a KP solution for $u$ which also satisfies the Toda lattice equation given
in the right hand side (see e.g. \cite{BK:03}).

\section{Appendix A: The dKP hierarchy}
The purpose of this Appendix is to show how one can get the kernel formula (\ref{Cauchykernel}) for the
KP hierarchy, which is one of the most important formulae in this paper. We begin by giving a quick introduction to the dKP hierarchy:
Let $\lambda$ be a Laurent  series in $p$ given by
\[
\lambda=p+\frac{u_2}{p}+\frac{u_3}{p^2}+\cdots,
\]
where $u_k$ are the functions of an infinite variables $(T_1,T_2,\ldots)$.
Then the dKP hierarchy is defined by the set of equations,
\[
\frac{\partial\lambda}{\partial T_n}=\{\Phi_n,\lambda\}\qquad n=1,2,\ldots,
\]
where $\Phi_n(\lambda(p))$ is the Faber polynomial given by
\[
\Phi_n(\lambda(p))=[\lambda(p)^n]_{\ge0}.
\]
Here $[f(p)]_{\ge0}$ is a polynomial part of the Laurent series $f(p)$ in $p$. In particular,
$\Phi_1(\lambda(p))=p$, and we denote sometime $T_1=X$.  The bracket $\{f,g\}$ is
the Poisson bracket defined by
\[
\{f,g\}=\frac{\partial f}{\partial p}\frac{\partial g}{\partial X}-\frac{\partial f}{\partial X}\frac{\partial g}{\partial p}\,.
\]
The Faber polynomial can be written in the Laurent series (see Definition \ref{faberFG}),
\[
\Phi_n(\lambda)=\lambda^n-\sum_{m=1}^{\infty}\frac{1}{\lambda^m}Q_{nm}\,.
\]
In particular, we have
\[
\Phi_1(\lambda(p))=p(\lambda)=\lambda^n-\sum_{m=1}^{\infty}\frac{1}{\lambda^m}Q_{1m}.
\]
Then, following the computation in \cite{CK:95} (see also \cite{CT:06}), we first multiply $\lambda^{n-1}\frac{\partial\lambda}{\partial p}$ to $p(\lambda)$, i.e.
\[
p(\lambda)\lambda^{n-1}\frac{\partial\lambda}{\partial p}
=\lambda^n\frac{\partial\lambda}{\partial p}-\sum_{m=1}^{\infty}{Q_{1m}}\lambda^{n-m-1}\frac{\partial \lambda}{\partial p}\,.
\]
Taking the projection of this equation on the polynomial part,
\[
p\frac{\partial}{\partial p}\left(\frac{1}{n}\Phi_n\right)=\frac{\partial}{\partial p}\left(\frac{1}{n+1}\Phi_{n+1}\right)-\sum_{m=1}^{n-1}Q_{1m}\frac{\partial}{\partial p}\left(\frac{1}{n-m}\Phi_{n-m}\right).
\]
Now multiply by $\mu^{-n}$ and sum over $n\ge1$, we have
\begin{align*}
p\sum_{n=1}^{\infty}\frac{1}{\mu^n}\frac{\partial}{\partial p}\left(\frac{\Phi_n}{n}\right)
&=\sum_{n=1}^{\infty}\frac{1}{\mu^n}\frac{\partial}{\partial p}\left(\frac{\Phi_{n+1}}{n+1}\right)-\sum_{n=1}^{\infty}\sum_{m=1}^{n-1}\frac{Q_{1m}}{m\mu^m}\frac{1}{\mu^{n-m}}\frac{\partial}{\partial p}\left(\frac{\Phi_{n-m}}{n-m}\right)\\
&=\mu\sum_{n=1}^{\infty}\frac{1}{\mu^n}\frac{\partial}{\partial p}\left(\frac{\Phi_n}{n}\right)-1-\sum_{m=1}^{\infty}\frac{Q_{1m}}{\mu^m}\sum_{j=1}^{\infty}\frac{1}{\mu^{j}}\frac{\partial}{\partial p}\left(\frac{\Phi_j}{j}\right)\,.
\end{align*}
This then gives,
\begin{align*}
1&=\left(\mu-p-\sum_{m=1}^{\infty}\frac{Q_{1,m}}{\mu^m}\right)\sum_{n=1}^{\infty}\frac{1}{\mu^n}\frac{\partial}{\partial p}\left(\frac{\Phi_n}{n}\right)\\
&=(p(\mu)-p)\sum_{n=1}^{\infty}\frac{1}{\mu^n}\frac{\partial}{\partial p}\left(\frac{\Phi_n}{n}\right).
\end{align*}
Thus we have
\[
\frac{1}{p(\mu)-p}=\sum_{n=1}^{\infty}\frac{1}{\mu^n}\frac{\partial}{\partial p}\left(\frac{\Phi_n(\lambda(p))}{n}\right)\,.
\]
Integrating this with respect to $p$, and fixing the boundary condition at $\mu=\infty$, we obtain
\begin{align*}
\frac{\mu}{p(\mu)-p(\lambda)}&=\exp\left(\sum_{n=1}^{\infty}\frac{1}{n\mu^n}\Phi(\lambda)\right)\\
&=\exp\left[\sum_{n=1}^{\infty}\frac{1}{n\mu^n}\left(\lambda^n-\sum_{m=1}^{\infty}\frac{1}{\lambda^m}Q_{nm}\right)\right]\\
&=\frac{1}{1-\frac{\lambda}{\mu}}\exp\left(-\sum_{n=1}^{\infty}\sum_{m=1}^{\infty}\frac{1}{\lambda^m\mu^n}\frac{Q_{nm}}{n}\right)\,. \end{align*}
Then expressing the Grunsky coefficients in terms of the free energy, i.e.
$c_{nm}=\frac{1}{n}Q_{nm}=\frac{1}{nm}F_{nm}$, we obtain
 the dKP hierarchy (\ref{Cauchykernel}),
\[
\frac{p(\mu)-p(\lambda)}{\mu-\lambda}=e^{D(\lambda)D(\mu)F}\,.
\]
One should note here that the dispersionless Hirota equation for $F_{nm}$ is completely
determined by the function $\lambda=p+\sum_{m=1}^{\infty}\frac{u_{m+1}}{p^m}$ (spectral curve).
Then the dKP hierarchy defines an integrable deformation of the function $\lambda(p)$, and
as a result, the Grunsky coefficients can be expressed by the second derivatives of
the free energy with the deformation parameters.

We also remark that in \cite{T:07} (see also \cite{TT:95}), Takasaki derives this formula as a classical limit of the differential
Fay identity, which can be obtained from the bilinear identity \cite{DJKM:83},
\[
\oint_{z=\infty}\frac{dz}{2\pi i}\,e^{\xi(\mathbf{t}'-\mathbf{t},z)}\tau\left(\mathbf{t}'-[z^{-1}]\right)
\tau\left(\mathbf{t}+[z^{-1}]\right)=0\,.
\]
where $\xi(\mathbf{t},z)=\sum_{n=1}^{\infty}z^nt_n$ and $[z^{-1}]$ is given by
\[
[z^{-1}]=\left(\frac{1}{z},\,\frac{1}{2z^2},\,\frac{1}{3z^3},\,\ldots\right)\,.
\]
Takasaki also derives the dPfaff hierarchy from the classical limit of the differential Fay identity
obtained from the bilinear identity for the DKP hierarchy \cite{JM:83}. This is remarkable, since there seems to be no
proper classical limit for the Pfaff lattice hierarchy in the matrix representation given in \cite{AvM:02, KP:07}
(this is contrary to the Toda lattice hierarchy).

\section{Appendix B:  Computation of $C_0(T_0)$}
Here we give the computation of $C_0^{(\beta)}(T_0)$ for $\beta=1,2$ and $4$.
We use formula (3.3.10) given in \cite{M:2004}, reproduced here:
with
\[
Z_{n}^{(\beta)}\left(\frac{\beta}{2}{\lambda^2};\mathbf{0}\right) = \int_{\mathbb{R}}d\lambda_1 \dots \int_{\mathbb{R}}d\lambda_n \, \prod_{i<j} \big| \lambda_i - \lambda_j \big|^\beta \,\exp\left( - \frac{\beta}{2} \sum_{j=1}^n \lambda_j^2 \right)\,,
\]
one finds the following formula, which we call Mehta's formula,
\[ Z_n^{(\beta)} \left(\frac{\beta}{2}\lambda^2;\mathbf{0}\right)= (2\pi)^{\frac{n}{2}} \beta^{-\frac{n}{2} - \beta \frac{n(n-1)}{4}} \left[ \Gamma\left( 1 + \frac{\beta}{2}\right) \right]^{-n} \prod_{j=1}^n \Gamma\left( 1 + \frac{\beta  j}{2}\right) \,.\]

\subsection{GUE case}
For $\beta=2$, we have
\[
 C^{(2)}(T_0;N) = \log\left[ Z^{(2)}_n\left( \frac{N}{2} \lambda^2; \mathbf{0}\right) \right]. \]
 From Mehta's formula, we have
\[
Z^{(2)}_n( \lambda^2; \mathbf{0}) = (2\pi)^{\frac{n}{2}} 2^{-\frac{n^2}{2}} \prod_{j=1}^n j!\,.
\]
A straightforward change of variables gives
\[
Z^{(2)}_n\left( \frac{N}{2} \lambda^2; \mathbf{0}\right) = N^{-\frac{n^2}{2}} (2\pi)^{\frac{n}{2}} \prod_{j=1}^n j!\,.
\]
Inserting this into our formula for $C(T_0)$ we find:
\[
\log\left[ Z^{(2)}_n\left(\frac{n}{2}\lambda^2;0\right)\right] = -\frac{n^2}{2} \log N+ \frac{n}{2} \log(2\pi)
+ \sum_{j=1}^n \log(j!)\,.
\]
So the problem is to find the leading order behavior of $\sum \log(j!)$.  We compute this in a generalization of the Riemann sums approach to deriving Sterlings approximation.

We start by rewriting the expression $\sum \log(j!)$ as a double summation which then simplifies:
\begin{align*}
\sum_{j=1}^n \log(j!) &= \sum_{j=1}^n \sum_{k=1}^j \log k= \sum_{k=1}^n \sum_{j=k}^n \log k \\
&= \sum_{k=1}^n (n-k+1) \log k\,.
\end{align*}
We then approximate this sum by an integral using the trapezoid rule
\[ \sum_{j=1}^n \log(j!) = \frac{1}{2} \log n + \int_1^n (n-x+1) \log x \,dx+ E\,,
\]
where the error in the approximation can be computed using the Euler-MacLaurin formula.  For our purposes it is enough to show that $E = \mathcal{O}(n\log n)$ in truth it performs much better.  The integral then gives
\[ \sum_{j=1}^n \log(j!) = - \frac{3}{4} n^2 + \frac{1}{2} n^2 \log n + \mathcal{O}(n\log n)\,. \]
Inserting this back into our expression for $C^{(2)}(T_0;N)$ we find
\[ C^{(2)}(T_0;N) = \frac{1}{2} n^2 \log n - \frac{1}{2} n^2 \log N + \frac{3}{4} n^2 + \mathcal{O}(n\log n) \,.\]
Then in the limit, we have with $n=NT_0$,
\[ C_0^{(2)}(T_0) :=\lim_{N\to\infty}\frac{1}{N^2}C(T_0;N)= \frac{1}{2} T_0^2 \log(T_0) - \frac{3}{4} T_0^2 \,. \]

\subsection{GOE case}
In the case of $\beta=1$, we have
\[C^{(1)}(T_0;N) = \log\left[ Z_{2n}^{(1)}\left( \frac{N}{4} \lambda^2; \mathbf{0}\right) \right]
\]
Mehta's formula then gives
\begin{align*}
Z_{2n}^{(1)}\left( \frac{1}{2}\lambda^2; \mathbf{0}\right) &= (2\pi)^n \Gamma\left(\frac{3}{2}\right)^{-2n} \prod_{j=1}^{2n} \Gamma\left(1 + \frac{j}{2}\right) \\
&= 2^{3n} \prod_{j=1}^{n} j! \prod_{j=0}^{n-1} \frac{\sqrt{\pi}}{2^{2j+1} } \frac{(2j+1)!}{j!}\\
&= 2^{2n} \pi^{\frac{n}{2}} n! \prod_{j=1}^{n-1} 2^{-2j} (2j+1)!\,.
\end{align*}
Likewise in this case we find that
\begin{align*}
 C^{(1)}(T_0;N) &= \log\left[ Z_{2n}^{(1)}\left( \frac{N}{4} \lambda^2; \mathbf{0}\right) \right] = \left(n^2 + \frac{n}{2}\right ) \log\frac{2}{N} + \log\left[ Z_{2n}^{(1)}\left(\frac{1}{2} \lambda^2; \mathbf{0}\right)\right]
\\
&=  n^2 \log \frac{2}{N} + \log(n!) - 2 \log 2\left( \sum_{j=1}^n j\right) + \sum_{j=1}^{n-1} \log[(2j+1)!] +\mathcal{O}(n\log n)\,.
\end{align*}
Stirlings approximation implies that
\[ \log(n!) = \sum_{k=1}^n \log k = \mathcal{O}(n \log n)\,. \]
So we only need to find the behavior of $\sum \log((2j+1)!) $ for large $n$: as above we follow the basic method of Stirlings approximation and write the sum as the approximation of an integral plus an error term.  For large $n$ the error term is of order $\mathcal{O}(1)$;
\begin{align*}
\sum_{j=1}^{n-1} \log((2j+1)!) &= \sum_{j=1}^{n-1} \sum_{k=1}^{2j+1} \log k= \sum_{k=1}^{2n-1} \sum_{j=\lfloor \frac{k}{2} \rfloor}^{n-1} \log k \\ \nonumber
&= \sum_{l=1}^{n-1}\left[ (n-l) \log(2l) + (n-l) \log(2l+1)\right] \\ \nonumber
&= \frac{1}{2} \log( (2n-2)(2n-1)) + \int_1^{n-1} (n-x) \log( 2x (2x+1)) dx + E \\ \label{gen-stirling}
&= -\frac{3}{2} n^2 + n^2 \log(2n) + \mathcal{O}(n\log n)\,. \end{align*}
Then we find
\[ C^{(1)}(T_0;N) = n^2 \log(2 n) - n^2 \log N - \frac{3}{2}n^2 \mathcal{O}(n\log n)\,, \]
which gives
\[  C_0^{(1)}(T_0) :=\lim_{N\to\infty}\frac{1}{N^2}C^{(1)}(T_0;N)= T_0^2 \log(2 T_0) - \frac{3}{2} T_0^2 \,. \]

\subsection{GSE case}
For $\beta=4$,
we have
\[ C^{(4)}(T_0;N) = \log\left[ Z_{n}^{(4)}\left( \frac{N}{4} \lambda^2; \mathbf{0}\right) \right]\,. \]
Mehta's formula with $\beta=4$ gives
\[ Z_n^{(4)}( \lambda^2 ;\mathbf{0}) = \pi^{\frac{n}{2}} 2^{-2n^2 + \frac{n}{2}} \prod_{j=1}^n (2j)! \,,\]
which together with a change of variables gives
\[ Z_n^{(4)}\left(\frac{N}{4} \lambda^2; \mathbf{0}\right) = \left(\frac{\pi}{2}\right)^{\frac{n}{2}}N^{-n^2 + \frac{n}{2}} \prod_{j=1}^n (2j)!\,. \]
Therefore
\[ C^{(4)}(T_0;N) = - n^2 \log N  + \sum_{j=1}^n \log(2j)! +\mathcal{O}(n\log n) \,.\]
So we must compute the asymptotic behavior of $\sum \log(2j)! $ it is not
suprising that it is identical to the asymptotic behavior of $\sum \log(2j+1)!
$ found in the $\beta=1$ case.  This is in fact yet more evidence of the
duality that exists between $\beta=1$ and $4$.
Inserting this back into our expression for $C^{(4)}(T_0;N)$ we find
\[ C^{(4)}(T_0;N) = - n^2 \log N - \frac{3}{2} n^2 + n^2 \log(2n) + \mathcal{O}(n \log n)\,,\]
from which we conclude that
\[  C_0^{(4)}(T_0) = T_0^2 \log(2 T_0) -\frac{3}{2} T_0^2.  \]

\medskip


\begin{thebibliography}{99}

\bibitem{AvM:99}
   M. Adler and P. van Moerbeke,
   Vertex operator solutions of the discrete KP-hierarchy,
   {\it Comm. Math. Phys.}, {\bf 203} (1999), 185-210.

\bibitem{AvM:02}
    M. Adler and P. van Moerbeke,
    Toda versus Pfaff lattice and related polynomials,
    {\it Duke Math. Journal}, {\bf 112} (2002) 1-58.

\bibitem{ASvM:02}
   M. Adler, T. Shiota and P. van Moerbeke,
   Pfaff $\tau$-functions,
   {\it Math. Ann.}, {\bf 322} (2002) 423-476.

\bibitem{AK:96}
   S. Aoyama and Y. Kodama,
   Topological Landau-Ginzburg theory with a rational potential and the dispersionless KP hierarchy,
   {\it Comm. Math. Phys.}, {\bf  182} (1996)  185-219.

\bibitem{BIZ}
        D. Bessis, C. Itzykson, and J.B. Zuber, Quantum field theory techniques in graphical enumeration, {\it Adv. Applied Math.} {\bf 1} (1980) 109-157.



\bibitem{BK:03}
    G. Biondini and Y. Kodama,
     On a family of solutions of the Kadomtsev-Petviashvili equation which also satisfy
    the Toda lattice hierarchy,
     {\it J. Phys. A: Math. Gen.} {\bf 36} (2003) 10519-10536.

\bibitem{BP:08}
W. Bryc and V. U. Pierce,
Duality of real and quaternionic random matrices, 0806.3695 (2008).

\bibitem{CT:06}
    Y.-T. Chen and M.-H. Tu,
    On kernel formulas and dispersionless Hirota equations,
    (arXiv:nlin/0605042).

\bibitem{CK:95}
    R. Caroll and Y. Kodama,
    Solution of the dispersionless Hirota equations,
   {\it J. Phys. A: Math. Gen.} {\bf 28} (1995) 6373-6387.


\bibitem{CGM-S:08}
    B. Collins, A. Guionnet, and E. Maurel-Segala, Asymptotics of unitary and orthogonal matrix integrals,  (arXiv:0608.193)

\bibitem{DJKM:83}
   E. Date, M. Kashiwara, M. Jimbo and T. Miwa,
   Transformation groups for soliton equations, in {\it Nonlinear Integrable Systems - Classical
   Theory and Quantum Theory} (World Scientific, 1983), 39-119.

\bibitem{DG:2007A}
P. Deift and D. Gioev, Universality in random matrix theory for orthogonal and symplectic ensembles. {\it IMRP } {\bf 2} (2007) 116.

\bibitem{DG:2007B}
P. Deift and D. Gioev, Universality at the edge of the spectrum for unitary, orthogonal, and symplectic ensembles of random matrices.  {\it Comm. Pure Appl. Math.} {\bf 60} (2007) 867-910.

\bibitem{D:92}
   B. A. Dubrovin,
   Integrable systems and classification of 2-dimensional topological field theories, In {\it Integrable Systems}, Prog. Math. 115, (Birkh\"auser, Boston, 1993), 313-359.

\bibitem{EM-2003}
   N. M. Ercolani and K. T.-R. McLaughlin, Asymptotics of the partition function for random matrices via Riemann-Hilbert techniques, and applications to graphical enumeration, {\it IMRN} {\bf 14} (2003) 755-820.

\bibitem{EMP-2008}
N. M. Ercolani, K. T.-R. McLaughlin, and V. U. Pierce, Random matrices, graphical enumeration and the continuum limit of Toda lattices.   {\it Comm. Math. Phys.}  278  (2008),  31-81.


\bibitem{GHJ}
  I. P. Goulden, J. L. Harer, and D. M. Jackson, A geometric parametrization for the virtual Euler characteristics of the moduli spaces of real and complex algebraic curves, {\it Trans. Amer. Math. Soc.} {\bf 353} (2001) 4405-4427.

\bibitem{HO:91}
R. Hirota and Y. Ohta,
Hierarchies of coupled soliton equations. I,
{\it J. Phys. Soc. Japan}, {\bf 60} (1991) 798-809.

\bibitem{J:1993}
D. M. Jackson, On an integral representation for the genus series for 2-cell
embeddings, {\it Trans. Amer. Math. Soc.}{\bf 344} (1993) 755-772.

\bibitem{JM:83}
M. Jimbo and T. Miwa,
Soliton equations and infinite dimensional Lie algebras,
{\it Publ. RIMS}, Kyoto University, {\bf 19} (1983) 943-1001.

\bibitem{K:99}
   S. Kakei,
   Orthogonal and symplectic matrix integrals and coupled KP hierarchy,
   {\it J. Phys. Soc. Japan}, {\bf 68} (1999) 2875-2877.

\bibitem{KP:07}
     Y. Kodama and V. U. Pierce,
     Geometry of the Pfaff lattices,
     {\it Inter. Math. Res. Notes}, (2007) rnm 120, 55 pages

\bibitem{KP:08}
     Y. Kodama and V. U. Pierce,
     The Pfaff lattice and the symplectic eigenvalue problem,
     (arXiv:0802.2288)

\bibitem{Kr:94}
    I. M. Krichever,
    The $\tau$-function of the unversal Whitham hierarchy, matrix models and topological field theories,
    {\it Comm. Pure Appl. Math.}, {\bf 47} (1994) 437-475.

\bibitem{KL:93}
J. H. Kwak and J. Lee, Genus polynomials of dipoles, {\it Kyungpook Math. J.} {\bf 33} (1993) 115-125.

\bibitem{MWZ:02}
   A. Marshakov, P. Wiegmann and A. Zabrodin,
   Integrable structure of the Dirichlet boundary proble in two dimensions,
   {\it Comm. Math. Phys.}, {\bf 227} (2002) 131-153.

\bibitem{M:2004}
M. L. Mehta, {\it Random Matrices 3rd ed.} Elseier, Amsterdam (2004).


\bibitem{MW-2003}
   M. Mulase and A. Waldron, Duality of orthogonal and symplectic matrix integrals and quaternionic Feynman graphs, {\it Comm. Math. Phys.} {\bf 240} (2003) 553-586.

\bibitem{S:45}
    I. Schur,
    On Faber polynomials,
    {\it Amer. J. Math.}, {\bf 67} (1945) 33-41.

\bibitem{S:07}
M. Shcherbina, On the universality for orthogonal ensembles of random matrice.  {arXiv:0701.01046}.

\bibitem{T:07}
    K. Takasaki,
    Differential Fay identities and auxiliary linear problem of integrable hierarchies.
    (arXiv:0710.5356).

\bibitem{TT:95}
    K. Takasaki and T. Takebe,
    Integrable hierarchies and dispersionless limit,
    {\it Rev. Mod. Pfys.}, {\bf 7} (1995) 743-808.

\bibitem{LT:03}
    L-P Teo,
    Analytic functions and integrable hierarchies - Characterization of tau functions,
    {\it Lett. Math. Phys.}, {\bf 64} (2003) 75-92.

\bibitem{WZ:00}
     P. B. Wiegmann and A. Zabrodin,
     Conformal maps and integrable hierarchies,
     {\it Comm. Math. Phys.}, {\bf 213} (2000) 523-538.

 \bibitem{Z:01}
     A. Zabrodin,
     The dispersionless limit of the Hirota equations in some problems of complex analysis,
     {\it Teoret. Mat. Fiz.} {\bf 129}, (2001) 239-257.


\end{thebibliography}
\end{document}